\documentclass[aps,notitlepage,twocolumn,superscriptaddress,nofootinbib,pra,10pt]{revtex4-1}
\usepackage{amsmath,amsfonts,amssymb,amsthm,graphicx,enumerate,times,color,bm}
\usepackage{graphicx}
\usepackage{listings}
\usepackage{tikz}
\usepackage{newfloat}
\usepackage{mathtools}
\usepackage{enumitem}
\usepackage{bbm}
\usepackage{hyperref}
\mathtoolsset{showonlyrefs}
\usepackage[super]{nth}
\DeclareFloatingEnvironment[fileext=frm,placement={!ht},name=Box]{myfloat}

\definecolor{beige}{rgb}{1.00,0.95,0.90}
\usepackage[framemethod=TikZ]{mdframed} 
\newenvironment{frameenv}[1]
    {\begin{myfloat}[]
    \begin{mdframed}[roundcorner=4pt,backgroundcolor=cyan!5]
    \caption{\bfseries #1}
    }
    { \vspace{0.3cm}
    \end{mdframed}\end{myfloat}
    }

\newtheorem{theorem}{Theorem}

\newtheorem{result}{Result}

\newtheorem{lemma}[theorem]{Lemma}
\newtheorem{relemma}[theorem]{Lemma}
\newtheorem{corollary}[theorem]{Corollary}
\newtheorem{definition}[theorem]{Definition}

\newtheorem{remark}[theorem]{Remark}

\newcommand{\RR}{\mathbb{R}}

\newcommand{\mc}[1]{\mathcal{#1}}

\newcommand{\ket}[1]{|#1\rangle}

\newcommand{\ketbra}[2]{| #1 \rangle \langle #2 |}

\newcommand{\proj}[1]{\vert #1\rangle\!\langle#1 \vert}

\def\id{{\mathbb I}}
\newcommand{\norm}[1]{\left\Vert #1 \right\Vert}
\def\one{\mathbbm{1}}

\newcommand{\Tr}{\operatorname{tr}}
\newcommand{\tr}{\Tr}

\begin{document}
\newcommand{\fu}{Dahlem Center for Complex Quantum Systems, Freie Universit{\"a}t Berlin, 14195 Berlin, Germany}
\author{Paul Boes }
\email{pboes@zedat.fu-berlin.de}
\affiliation{\fu}
\author{Nelly H. Y. Ng}
\email{nelly.hy.ng@gmail.com}
\affiliation{School of Physical and Mathematical Sciences, Nanyang Technological University, 639673, Singapore}
\affiliation{\fu}
\author{Henrik Wilming}
\email{henrik-physics@arrr.de}
\affiliation{Institute  for  Theoretical  Physics,  ETH  Zurich,  8093  Zurich,  Switzerland}
\affiliation{Leibniz Universit\"at Hannover, Appelstra\ss e 2, 30167 Hannover, Germany}
\title{The variance of relative surprisal as single-shot quantifier}

\date{\today}

\begin{abstract}
    The variance of (relative) surprisal, also known as varentropy, so far mostly plays a role in information theory as quantifying the leading order corrections to asymptotic i.i.d.~limits. Here, we comprehensively study the use of it to derive single-shot results in (quantum) information theory. We show that it gives genuine sufficient and necessary conditions for approximate state-transitions between pairs of quantum states in the single-shot setting, without the need for further optimization. We also clarify its relation to smoothed min- and max-entropies, and construct a monotone for resource theories using only the standard (relative) entropy and variance of (relative) surprisal. This immediately gives rise to enhanced lower bounds for entropy production in random processes.
    We establish certain properties of the variance of relative surprisal which will be useful for further investigations, such as uniform continuity and upper bounds on the violation of sub-additivity. Motivated by our results, we further derive a simple and physically appealing axiomatic single-shot characterization of (relative) entropy which we believe to be of independent interest. We illustrate our results with several applications, ranging from interconvertibility of ergodic states, over Landauer erasure to a bound on the necessary dimension of the catalyst for catalytic state transitions and Boltzmann's H-theorem.
\end{abstract}

\maketitle

\section{Introduction} 
\label{sec:introduction}
Many central results of quantum information theory are concerned with the manipulation of quantum systems in the so-called asymptotic i.i.d.~setting, in which one considers the limit of taking infinitely many independent and identically distributed copies of a quantum system~\cite{schumacher1995quantum,Jozsa1994,Bennett1996,Schumacher1997,Holevo1998,lloyd1997capacity,hayden2001asymptotic,PhysRevLett.111.250404,PhysRevA.67.062104}. While not always being physically realistic, 
this setting is convenient to work in because it allows for the application of standard concentration results from statistics and information theory~\cite{hoeffding1952large,cover1999elements}. Recent years have seen a lot of effort in studying more general settings, in which subsystems might be correlated with one another or the size of the system is finite. Conceptually, the most extreme weakening of the i.i.d.~setting is the \emph{single-shot} setting, in which generally no assumptions about the size of the system or its correlations are made. This setting derives its name from the fact that it can be seen to describe a single iteration of a protocol, in contrast to the i.i.d.~setting which is concerned with infinitely many independent iterations.  Several classic results~\cite{hardy1929some,blackwell1953,ruch1978mixing,RUCH1980222,Marshall2011} have been instrumental to the recent developments of majorization-based resource theories~\cite{PhysRevA.95.012110,wang2019resource,PhysRevLett.111.250404,horodecki2013fundamental,renes2016relative,PhysRevLett.116.120404,gour2017quantum,chitambar2019quantum}, which have found diverse applications in entanglement theory, thermodynamics, asymmetry and many more central topics in information theory.
By now there exists a detailed understanding of an intuitive trade-off between the above settings (and that we detail in later sections).  On the one hand, the i.i.d.~setting, due to its various assumptions, can usually be characterized by variants of the (quantum) relative entropy, 
\begin{equation}\label{eq:def_relent}
     {S(\rho \| \sigma) := \tr(\rho (\log(\rho) - \log(\sigma))}.
\end{equation}
 {On the other hand, majorization-like conditions play a central role for single-shot transformations, even extending to the study of approximate transitions. Many such results are tight in the  {quasi-}classical regime~\cite{VanderMeer2017,horodecki2018extremal} (when we restrict to transformations between pairs of states that commute with another). Nevertheless, the number of conditions that need to be checked increases linearly with the dimension of the states of interest, which 
makes a more systematic and simplified understanding in characterizing the possibility of state transitions difficult.} 
Fortunately, the family of \emph{smoothed entropies} have turned out to be a powerful tool to describe these constraints and operationally characterize a variety of single-shot tasks~\cite{Renner2004,Renner2005,Renner2005a,Renner2006,konig2009operational,Buscemi2009,tomamichel2010duality,Brandao2011,Wang2013,Datta2013,Datta2013a,Brandao2015,Tomamichel2016a,Anshu2017,gour2017quantum}. However, they typically require an optimization over the potentially high-dimensional state space.

In this work, we develop a  complementary approach to the study of the single-shot setting in which ``single-shot effects'' are witnessed and quantified by a single quantity, the \emph{variance of (relative) surprisal},
\begin{equation}\label{eq:def_variance}
    V(\rho \| \sigma ):= \tr \left(\rho\log^2\left(\frac{\rho}{\sigma}\right)\right)-S(\rho\| \sigma)^2,
\end{equation}
where $\log(\frac{\rho}{\sigma}) \equiv \log \rho - \log \sigma$ is the \emph{relative surprisal}  {of one quantum state $\rho$ with respect to another $\sigma$}, and we use $\log\equiv\log_2$. In the following, we refer to this quantity simply as the relative variance. The relative variance has been shown to measure leading-order corrections to asymptotic results in (quantum) information theory \cite{Strassen1962,Hayashi2008,Hayashi2009,Polyanskiy2010,Tomamichel2012,Verdu2012,Altug2014,Li_2014,datta2014second,Tan2015,Tomamichel2015,Tomamichel2016,kumagai2016second,chubb2018beyond,Chubb2018,Chubb2019,korzekwa2019avoiding}, but here we show that its role extends to the genuine single-shot setting. We do this from two points of view: First, we show that the relative variance quantifies single-shot corrections to possible state-transitions between pairs of  quantum states. These imply that state-transitions between pairs of states with low relative variance are essentially characterized by the relative entropy and hence do not exhibit strong single-shot effects. Simple examples of such states are ergodic states (see below).   

Since single-shot effects often represent operational impediments (such as, for example, decreased ability to extract work in the context of thermodynamics), the above
finding motivates the question whether there is, in some sense, a ``cost'' associated to obtaining states of low relative variance. We proceed to show that this is indeed the case: Reducing the relative variance between a pair of states necessitates a proportional reduction of the relative entropy between those states. In the resource-theoretic setting (defined later), in which the relative entropy is known to itself be measure of a state's resourcefulness, this finding implies that increasing the operational value of a system in the sense of reducing its relative variance comes at the cost of reducing its value as measured by the relative entropy. Formally, the above tradeoff result is an implication of a resource monotone that we construct. We formulate the trade-off between relative variance and relative entropy both for single systems as well as for the \emph{marginal} changes of bipartite systems.

Overall, our findings motivate the relative variance as a simple measure of single-shot effects: The smaller the relative variance of pairs of states, the better their single-shot properties are described by the relative entropy, and vice versa. In addition to the above results, we clarify the relation between the relative variance and the smoothed entropies, relate our single-shot results to the known asymptotic leading order corrections mentioned above, and present a novel axiomatic characterization of the (relative) entropy from a single, physically motivated axiom. We believe the latter finding to be of independent interest. 
Apart from its applications in the context of resource theories, our work can also be seen as a thorough investigation of the mathematical properties of the variance of relative surprisal, such as uniform continuity bounds, corrections to sub-additivity and the relation of the variance of relative surprisal to smoothed R\'enyi divergences and the theory of approximate majorization.

The remainder of this paper is structured as follows:  Our results are concerned with generic state transitions between pairs of states in the quasi-classical setting. However, following a formal introduction of the setup and terminology below, 
for the sake of clarity, in Section~\ref{sec:overview_of_main_results} we first provide an overview of our main results for the special case of state transitions under unital channels, which also corresponds to the resource theory of purity. The results are then shown in full generality and further discussed in Section~\ref{sec:relative_results}. Throughout, we focus on the formal results and provide applications mostly for illustration in boxes, leaving a more detailed study of applications and implications of our results to future work.

\subsection{Setup and Notation} 
\label{sub:setup_and_notation}

Let $S$ and $S'$ be quantum systems represented on Hilbert spaces with fixed and finite respective dimensions $d$ and $d'$. Further, let $\mc D(S)$ be the set of quantum states on $S$, and similarly for $S'$. 
Given two pairs of states $\rho, \sigma \in \mc D(S)$ and $\rho', \sigma' \in \mc D(S')$ --- sometimes referred to as dichotomies --- we write $(\rho, \sigma) \succ (\rho', \sigma')$ if there exists a quantum channel $\mc E: \mc D(S) \to \mc D(S')$ such that $\mc E(\rho) = \rho'$ and $\mc E(\sigma) = \sigma'$.
The pre-order $\succ$ has been studied both classically and quantumly (e.g. \cite{blackwell1953,ALBERTI1980163,PhysRevA.95.012110,Buscemi2019information}) and forms the backbone of many \emph{resource theories}. 
In a resource theory a set of so-called \emph{free states} $\mc{D}_{\mathcal{F}} \subseteq \mathcal{D}(S)$ is specified, together with a set of quantum channels $\mathcal{F}$ such that each of these channels maps free states into free states, $\mathcal{F}[\mathcal{D}_{\mathcal{F}}] \subseteq \mathcal{D}_{\mathcal{F}}$. 
These channels are therefore called \emph{free}. This terminology originates from the idea that free states constitute a class of states that are easy to prepare in a given physical context, while free channels constitute physical operations that are easy to implement in this context (see Ref.~\cite{chitambar2019quantum} for a review on resource theories).
One is then concerned with the special case of the ordering $\succ$ in which $\mathcal{E}$ is free. 
In the special case that there is only a single free state $\sigma$, this corresponds to the special case $\sigma = \sigma'$ and induces the pre-order $\succ_\sigma$ on $\mc D(S)$, known as $\sigma$-majorization, where $\rho \succ_\sigma \rho'$ is equivalent to $(\rho, \sigma) \succ (\rho', \sigma)$. 
As such, the results presented below can naturally be applied to those resource theories that follow the above form, although we emphasize that they hold more generally as well.

 {Since we are often concerned with approximate state transitions, we further write $(\rho, \sigma) \succ_\epsilon (\rho', \sigma')$ whenever there exists a state $\rho_\epsilon ' $ such that $(\rho, \sigma) \succ (\rho'_\epsilon, \sigma')$ and $D(\rho', \rho_\epsilon') := \frac{1}{2}\norm{\rho' - \rho_\epsilon'}_1 \leq \epsilon$.}

 Finally, the i.i.d.~setting corresponds to the special case of dichotomies of the form $(\rho^{\otimes n}, \sigma^{\otimes n}) \succ ({\rho'}^{\otimes n}, {\sigma'}^{\otimes n})$ for some $n \in \mathbb{N}$. Here, one usually considers the asymptotic limit $n \to \infty$. The statement we made in the introduction, that in this limit the relative entropy characterizes possible state transitions, is based on the well-known fact that, for any two dichotomies $(\rho, \sigma)$ and $(\rho', \sigma')$ it holds that, for any $\epsilon > 0$, there exists a sufficiently large $n$ such that $(\rho^{\otimes n}, \sigma^{\otimes n}) \succ_\epsilon ({\rho'}^{\otimes n}, {\sigma'}^{\otimes n})$ if and only if $S(\rho \| \sigma) > S(\rho' \| \sigma')$. We will recover this statement as a special case from results below.

\section{Overview of main results} 
\label{sec:overview_of_main_results}

 In this section, we provide an overview of our results for the special case $\succ_{\id}$, where $$\id \equiv \one_d/d$$ is our shorthand for the maximally mixed state in $d$ dimensions, whereas $\one_d$ is the $d$-dimensional identity operator.
    This choice of $\sigma$ corresponds to the resource theory of purity, also known as the resource theory of stochastic non-equilibrium~\cite{Gour2015}, and captures the essential insights from our results while being easier to state.

Channels that preserve the maximally mixed state are also called \emph{unital channels}~\footnote{In general, unital channels are channels that map the identity of their domain to the identity of their image, and hence are more general than the channels we consider here. 
    We ignore this difference because here we are only interested in strict preservation of the input state.} and the ordering $\succ_{\id}$
    is known as \emph{majorization}~\footnote{There exist various equivalent definitions of majorization between quantum states. At the level of $d$-dimensional probability distributions, we say that  $\vec{p}\succ_{\id} \vec{q}$ iff for all $k=1, 2, \dots, n-1$ it holds that$\sum_{i=1}^{k} p_i^{\downarrow}\geq \sum_{i=1}^{k} q_i^{\downarrow }$. An alternative definition to the one given above is to say that that $\rho$ majorizes $\rho'$ when the vector of eigenvalues of $\rho$ majorizes that of $\rho'$.}.
Unital channels are often used to model random processes. For instance, a (strict) subset of unital channels are random unitary channels that describe the evolution of a system under a unitary operator that was drawn at random from some fixed distribution. As such, the resource theory of purity is concerned with describing the evolution and operational ``value'' of quantum states in the presence of random processes.

\subsection{Variance of surprisal}
The central quantity in this work is the relative variance, defined in Eq.~\eqref{eq:def_variance}. For $\sigma = \id$, the relative variance reduces to the \emph{variance (of surprisal)}
\begin{equation}
    V(\rho) := V(\rho \| \id) = \tr(\rho\log^2(\rho)) - S(\rho)^2,
\end{equation}
where $S(\rho):=\log(d) - S(\rho \| \id) = -\tr(\rho \log(\rho))$ is the von Neumann entropy (itself the mean of the \emph{surprisal}, $-\log(\rho)$).
The variance of surprisal is also known as \emph{information variance} or \emph{varentropy}, and as \emph{capacity of entanglement} in the context of entanglement in many-body physics~\cite{Yao2010,Schliemann2011,Boer2019}.

Operationally, the variance can be understood, for instance, as the variance of the length of a codeword in an optimal quantum source code. However, while, as mentioned in the introduction, it is well known to quantify the leading order corrections to various (quantum) information theoretic tasks in the asymptotic limit,
its relevance in the single-shot setting has not yet, to the authors' knowledge, been thoroughly investigated (recently, some formal properties have, however, been developed in Ref.~\cite{Dupuis2019}). In this work, we study the (relative) variance and show that it in fact provides a useful measure of single-shot effects for approximate state transitions.

We begin by mentioning some properties that the variance of surprisal fulfills, and which we use throughout the paper:
\begin{enumerate}
    \itemsep0em
    \item Additivity under tensor products: 
          \begin{equation}
            V(\rho_1 \otimes\rho_2) = V(\rho_1)+V(\rho_2).
          \end{equation}
    \item Positivity: $ V(\rho)\geq 0 $.
    \item Uniform continuity (Lemma~\ref{lem:uniform contintuity relative variance}): 
          \begin{equation}
            |V(\rho) - V(\rho')|^2 \leq K \log^2(d) \cdot D(\rho, \rho')
          \end{equation}
          for a constant $K$.
    \item Correction to subadditivity (Lemma~\ref{lem:subadditivity_variance}):
          \begin{equation}
            V(\rho) \leq V(\rho_1) + V(\rho_2) +  K' \log^2(d) \cdot f(I_\rho), 
          \end{equation}
          for any bipartite state $\rho$ with respective marginal states $\rho_1$, $\rho_2$ and mutual information $I_\rho$, with constant $K'$ and $f(x) = \max\{\sqrt[4]{x}, x^2\}$.
    \item $V(\rho) = 0$ if and only if all non-zero eigenvalues of $\rho$ are the same. We call such states \emph{flat states}. Examples include any pure state and the maximally mixed state.
    \item For fixed dimension $d \geq 2$, the state $\hat{\rho}_d$ with maximal variance has the spectrum~\cite{Reeb2015}
          \begin{equation} \label{eq:sharp_state}
            \mathrm{spec}(\hat{\rho}_d) = \left(1 - r, \frac{r}{d-1}, \dots, \frac{r}{d-1}\right)
          \end{equation}
          with $r$ being the unique solution to 
          \begin{align}
            (1-2r)\ln\left( \frac{1-r}{r}(d-1)\right) = 2. 
          \end{align}
          We have $\frac{1}{4}\log^2(d-1) < V(\hat{\rho}_d) < \frac{1}{4}\log^2(d-1) + 1/\ln^2(2)$, and, in the limit of large $d$,  $r \approx \frac{1}{2}$. 
\end{enumerate}

Properties 3 and 4 are original contributions of this work and are part of our main technical results.

\subsection{Sufficient criteria for single-shot state transitions} 
\label{sub:single_shot_state_transitions}

It is well known that when considering the i.i.d.~limit, approximate majorization reduces to an ordering with respect to the von Neumann entropy.
More precisely, given two states $\rho$ and $\rho'$, then $S(\rho')>S(\rho)$ implies that for any $\epsilon>0$ there exists a number $N_\epsilon\in\mathbb N$ such that 
\begin{equation}\label{eq:iidsimple}
    \rho^{\otimes n}  {~\succ_{\id,\epsilon}}~ \rho'^{\otimes n} \quad \forall \: n \geq N_\epsilon.
\end{equation}
However, this is not the case in a single-shot setting. Here, the question whether $\rho \succ_{\id}\rho'$ in full generality depends on $d-1$ independent constraints on the spectra of $\rho$ and $\rho'$. This makes dealing with exact single-shot state transitions considerably more difficult.

Our first result shows that for approximate state transitions, there nevertheless exist simple sufficient conditions at the single-shot level that also involve the von Neumann entropy, but with a correction quantified by the variance:

\begin{result}[Sufficient conditions for approximate state transition] \label{res:suff_transition_unital}
    Let $\rho, \rho'$ be two states on $S$ and $1>\epsilon > 0$. If 
    \begin{equation}
        S(\rho') - \sqrt{V(\rho')(2\epsilon^{-1}-1)} ~\geq~ S(\rho) + \sqrt{V(\rho)(2\epsilon^{-1}-1)}, 
    \end{equation}
    then $\rho ~ {\succ_{\id,\epsilon}} ~ \rho'$.
\end{result}

We emphasize that this result is a fully single-shot result. It shows that the variance quantifies the single-shot deviation from the above i.i.d.~case. 
A convenient reformulation of Result 1 is as follows:
Let $\rho, \rho'$ be two states on $S$ with 
\begin{equation}
    S(\rho') - S(\rho) = \delta > 0.
\end{equation}
Solving for $\epsilon$ in Result~\ref{res:suff_transition_unital}, we see that
$   \rho~  {\succ_{\id,\epsilon}}~ \rho' $
where
\begin{equation}
    \epsilon \leq 
    \frac{2}{\delta^2}\left[\sqrt{V(\rho)} + \sqrt{V(\rho')}\right]^2, 
\end{equation}
can be achieved. An appealing feature of this result is that it does not require any optimization, as is typically present in results relying on smoothed entropies  {(compare, for example Result~\ref{res:suff_transition_unital} and its generalization Thm~\ref{res:suff_transition_relative} to the results in \cite{Buscemi2019information})}. 
When applied to state transitions under unital channels in the i.i.d.~limit, Result~\ref{res:suff_transition_unital} straightforwardly produces finite-size corrections towards asymptotic interconvertibility: given two states $\rho $ and $\rho'$ with $S(\rho') > S(\rho)$, it implies that $\rho^{\otimes n}  {~\succ_{\id,\epsilon}}~ (\rho')^{\otimes n}$ with 
\begin{equation} \label{eq:asymptotic_sufficient_bound}
    \epsilon \leq \frac{2[\sqrt{V(\rho)} + \sqrt{V(\rho')}]^2}{n [S(\rho') - S(\rho)]^2},
\end{equation}
which vanishes in the limit $n \to \infty$. For finitely many copies of the two states, the variances of initial and final states bound the achievable precision.

In Appendix~\ref{app:second-order asymptotics} we provide a more detailed analysis of the i.i.d.~case, where we use Result~\ref{res:suff_transition_unital} (and its  {generalization to $\sigma$-majorization}) to study convertibility between sequences of $n$ i.i.d.~states for large but finite $n$ and with an error $\epsilon_n$ such that $n\epsilon_n \rightarrow \infty$, but possibly $\epsilon_n\rightarrow 0$. This can be seen as a simple form of moderate-deviation analysis, and in particular we recover the ``resonance"-phenomenon reported in Refs.~\cite{Chubb2019,korzekwa2019avoiding}, namely that second-order corrections vanish when \begin{equation}
	\frac{V(\rho)/S(\rho)}{V(\rho')/S(\rho')} = 1,
\end{equation} 
with the proof being as simple as solving a quadratic equation.

Result~\ref{res:suff_transition_unital} implies that state transitions between pairs of initial and final states with low variance are essentially characterized by the entropy. As an application, in Box~\ref{box:ergodic}, we prove %
a simpler version of a recent result on the macroscopic interconvertibility of ergodic states under thermal operations~\cite{Faist2019,Sagawa2019}. 
Finally, let us note that Eq.~\eqref{eq:asymptotic_sufficient_bound} is in general not very tight, since we know from Hoeffding-type bounds that, in the i.i.d.~limit, the amount of probability outside of the typical window (which directly contributes to the error) scales as $\epsilon \propto \exp(-n)$. Nevertheless, the absence of any trailing terms makes it simple to evaluate, especially in single-shot scenarios. 

\begin{frameenv}{Interconvertibility of ergodic states}
    \label{box:ergodic}
    As an application of Result~\eqref{res:suff_transition_unital}, we discuss the interconvertibility of \emph{ergodic states} in the unital setting. 
    Informally speaking, ergodic states are states on an infinite chain of identical, finite-dimensional Hilbert-spaces of dimension $d$, enumerated by $\mathbb Z$ and called sites below, which have the property that correlations between observables located at distance sites converge to zero as their distance is increased. 
    See Ref.~\cite{Faist2019} for a detailed description of ergodic states. 
    Importantly,  states are in general correlated, examples being ground states of gapped, local Hamiltonians or thermal states of many-body systems away from the critical temperature. 
    Nevertheless, if $\rho_n$ denotes the density matrix of $n$ consecutive sites of the chain with entropy $S_n=S(\rho_n)$, the (quantum) Shannon-MacMillan-Breimann theorem shows that for arbitrarily small $\epsilon>0$ and sufficiently large $n$, we can find an approximation $\rho^\epsilon_n$ of $\rho_n$ with the property that each eigenvalue $p_j$ of $\rho^\epsilon_n$ fulfills~\cite{Bjelakovic2003}
    \begin{align}
        |-\log(p_j)-S_n| \leq 2n\epsilon 
    \end{align}
    and $D(\rho_n, \rho^\epsilon_n) \leq \epsilon$.
    Thus, the variance fulfills
    \begin{align}
        V(\rho^\epsilon_n) & =  \sum_j p_j \left(\log(1/p_j)-S_n^\epsilon\right)^2             \\
                           & =\sum_j p_j \left(\log(1/p_j)-S_n\right)^2 - (S_n-S_n^\epsilon)^2 \\
                           & \leq 4 n^2 \epsilon^2,                                            
    \end{align}
    where $S_n^\epsilon=S(\rho^\epsilon_n)$.
    By uniform continuity of the variance, Lemma~\ref{lem:uniform contintuity relative variance}, we therefore have 
    \begin{align}
        V(\rho_n) \leq V(\rho^\epsilon_n) + K n^2 \sqrt{\epsilon} \leq {\tilde K}^2 n^2 \sqrt{\epsilon}, 
    \end{align}
    with $\tilde K$ some constant and where we used $\epsilon^2\leq \sqrt{\epsilon}$. 
    Result~\ref{res:suff_transition_unital} now tells us that if we have two ergodic states with entropies $S_n=sn$ and $S'_n = s'n$ such that $s<s'$, then for any $\epsilon>0$ and sufficiently large $n$, we can convert $\rho_n$ to $\rho'_n$ using a unital channel with error at most
    \begin{align}
        2\left[\frac{\sqrt{V(\rho_n)} + \sqrt{V(\rho'_n)}}{(s'-s)n}\right]^2 \leq  \frac{16 {\tilde K}^2}{(s'-s)^2} \sqrt{\epsilon}, 
    \end{align}
    which can be made arbitrarily small. 
\end{frameenv}

\subsection{Relation to smoothed min- and max-entropies} 
\label{sub:relation_to_smooth_renyi_entropies}
As mentioned in the introduction, smoothed generalized entropies are often used to describe single-shot processes. For instance, the continuous family of \emph{R{\'e}nyi entropies} has been found to characterize possible single-shot transitions in the semi-classical setting~\cite{PhysRevA.64.042314,Brandao2015}. Among those entropies, the smoothed \emph{min-} and \emph{max}-entropies are of particular prominence, since they enjoy clear operational meanings in various information processing tasks such as randomness extraction or data compression (see, e.g., Refs.~\cite{Renner2004,konig2009operational}) and, in a sense, quantify complementary single-shot properties of quantum states. 
These quantities, the precise definition of which is given in Section~\ref{ssub:min_and_max_relativ_entropies}, can differ significantly from the von Neumann entropy for arbitrary states. However, 
the following result shows that one can bound this difference by the variance.

\begin{result}[Bounds on smoothed min- and max-entropies] \label{res:bound_min_max_by_variance}
    Let $1> \epsilon > 0$ and let $\rho$ be a state on $S$. Then, 
    \begin{align} \label{eq:bound_min_div_by_variance}
        S_{\max}^{\epsilon}(\rho) - S(\rho) & \leq \sqrt{( \epsilon^{-1} -1)V(\rho)} , \\
        S(\rho) - S_{\min}^{\epsilon}(\rho) & \leq \sqrt{( \epsilon^{-1} -1)V(\rho) }. 
    \end{align}
\end{result}

As a straightforward corollary, Result~\ref{res:bound_min_max_by_variance} implies that for any $1>\epsilon > 0$ and any state $\rho$, 
\begin{equation}\label{eq:upper_bound_on_Smin_Smax_diff}
    S_{\max}^{\epsilon}(\rho) - S_{\min}^{\epsilon}(\rho) \leq 2\sqrt{( \epsilon^{-1} -1)V(\rho)} .
\end{equation}
Result~\ref{res:bound_min_max_by_variance} has the appealing feature of providing an upper bound that factorizes into the variance and a function of the smoothing parameter $ \epsilon $. 
In analogy with Result~\ref{res:suff_transition_unital}, Result~\ref{res:bound_min_max_by_variance} shows that for states with small variance, finite-size corrections (e.g. to coding rates) are less pronounced, making these states ideal candidates for information encoding and transmission. This makes intuitive sense if we recall that the spectrum of states with zero variance is uniform over its support. Within the class of states with the same entropy, these \emph{flat} states are therefore maximally ``compressed'' and ``random''.
As a side-note, we observe that various models for ``batteries'' used in single-shot quantum thermodynamics restrict to battery states with zero or very low relative entropy variance. This choice can conveniently be interpreted in terms of Eq.~\eqref{eq:upper_bound_on_Smin_Smax_diff}, since $ S_{\rm min}^\epsilon $ and $ S_{\rm max}^\epsilon $ quantify the amount of single-shot work of creation and extractable work respectively for non-equilibrium states~\cite{horodecki2013fundamental,Brandao2015}. Restricting to low variance battery states therefore ensures that the process of storing and extracting work from a battery involves a minimal dissipation of heat.

\subsection{Decrease of variance} 
\label{sub:tradeoff_between_variance_decrease_and_entropy_production}

The previous two results establish the variance as a measure of single-shot effects, bounding the extent of such effects in the context of approximate state transitions and operational tasks such as data compression or randomness extraction. In particular, they imply that the manipulation of states with low variance produces less overhead due to finite-size effects. 
Therefore, state transitions between states of low variance can be operationally advantageous. This motivates the question whether there exists a resource-theoretic ``cost'' associated to decreasing the variance of a state. We show that this is indeed the case: Under unital channels, decreasing the variance lower bounds \emph{entropy production}. 

The statement above, formulated in Result~\ref{res:bound_entropy_production} is a consequence of a new resource monotone that we derive. 
\begin{definition}[Resource monotone]\label{def:RM}
     Let $\mc D$ be the set of all quantum states. A function $f: \mc D \times \mc D \rightarrow \RR$ is called a \emph{(resource) monotone} if for any dichotomies satisfying $(\rho, \sigma) \succ (\rho', \sigma')$, we also have $f(\rho, \sigma) \geq f(\rho', \sigma')$.
\end{definition}

Non-increasing resource monotones with respect to majorization are called \emph{Schur convex} and non-decreasing ones \emph{Schur concave}. Resource monotones are an important tool in the study of resource theories. For example, for a Schur convex function $f$, $f(\rho') > f(\rho)$ suffices to conclude that $\rho \nsucc_{\id} \rho'$ but the former is often far easier to check than the latter. An example of a Schur convex function is the purity $\tr(\rho^2)$ of a state $\rho$, while the entropy $S(\rho)$ is Schur concave. The variance itself is evidently not monotone, which might partially explain why it has so far not been studied resource-theoretically. However, the following lemma shows that the variance and entropy jointly give rise to a monotone.

\begin{lemma}[Schur-concavity of $M$] \label{res:monotone}
    The function 
    \begin{equation}\label{eq:monotone_definition_unital}
        M(\rho) :=  V(\rho) + \left(\frac{1}{\ln(2)} + S(\rho)\right)^2,
    \end{equation} 
    is Schur concave.
\end{lemma}

By Schur-concavity, we have $0 \leq M(\rho) \leq (\frac{1}{\ln(2)} + \log(d))^2$. Notably, unlike many commonly used monotones, $M$ is not additive with respect to product states. Fig.~\ref{fig:monotones} compares the regions of increasing $M$ and entropy compared to the majorization ordering for two initial states in $d=3$ and fixed eigenbasis, illustrating that for some states $M$ provides strictly stronger necessary conditions for the majorization ordering than the entropy $S$. 

By means of Lemma~\ref{res:monotone}, we can derive the following bound on entropy production, which is our third main result and proven in the more general statement of Corollary~\ref{cor:bound_relative_entropy_production}.
\begin{result}[Lower bound on entropy production] \label{res:bound_entropy_production}
    Let $\rho, \rho'$ be a pair of states on $S$. If $\rho \succ_{\id}\rho'$, then 
    \begin{equation} \label{eq:bound_entropy_production}
        S(\rho') - S(\rho) \geq \frac{V(\rho) - V(\rho')}{2 \sqrt{M(\rho)}} \geq \frac{V(\rho) - V(\rho')}{2(1/\ln(2) + \log(d))}.
    \end{equation}
\end{result}

\begin{figure}[h!]
    \includegraphics[width=.8\linewidth]{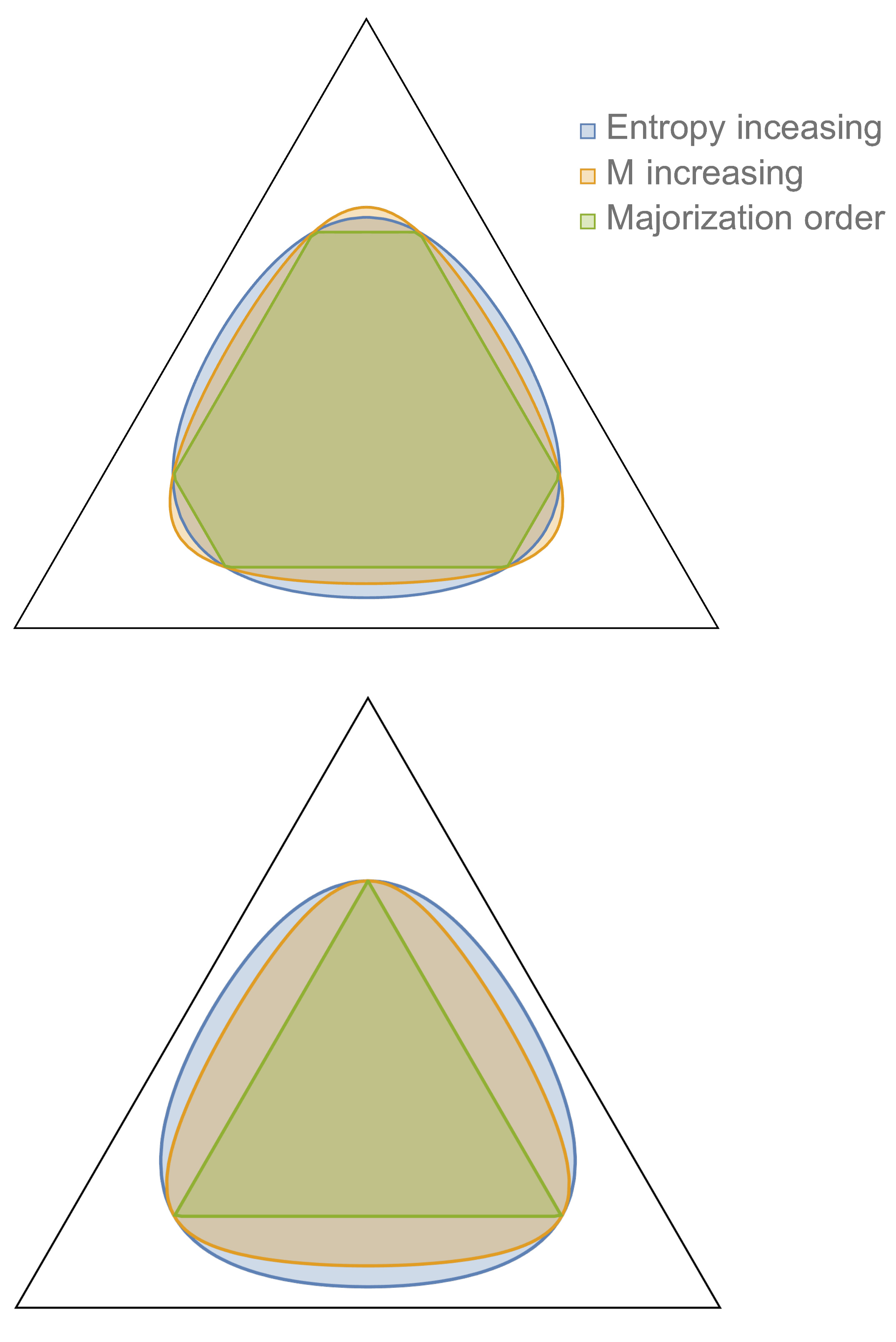}
    \caption{Region plots of increasing $M$ and entropy $S$ for two initial states (top: $p=(0.65,0.25,0.1)$; bottom: $p=(0.7,0.15,0.15)$) in the simplex of states $\rho = \sum_{i=1}^3 p_i \proj{i}$, i.e., for $d=3$ and fixed eigenbasis. In orange and blue are shown the sets of states with increasing $M$ and $ S $ respectively, while green shows the set of states majorized by the initial state. While in general neither red or blue region contains the other (top), for some states $M$ provides strictly stronger necessary conditions to rule out majorization (bottom).} 
    \label{fig:monotones}
\end{figure}

Result~\ref{res:bound_entropy_production} shows that the decrease of variance under unital channels can only come at the cost of increasing the system's entropy. It also complements the previous results, in that it shows that states with positive variance \emph{necessarily} exhibit finite-size effects: The entropy fails to characterize single-shot transitions, as witnessed by the variance. Still, the i.i.d.~limit is consistent with Eq.~\eqref{eq:bound_entropy_production}, since in the asymptotic limit the LHS grows linearly, while the other terms remain constant to leading order. At the same time, there exist sequences of state transitions for which the constraint imposed by Eq.~\eqref{eq:bound_entropy_production} remains non-trivial in the limit of large system size. An example are transitions from the state $\hat{\rho}_d$, as defined by Eq.~\eqref{eq:sharp_state}, to any state of constant variance, in the limit of growing $d$.
As a simple application, in Box~\ref{appl:erasure}, we apply Result~\ref{res:bound_entropy_production} to the task of erasure and find corrections to Landauer's principle that are quantified by the variance.

\begin{frameenv}{Finite-size corrections to Landauer's principle}\label{appl:erasure}
    Landauer's principle states that the erasure of information physically requires the dissipation of entropy which consumes work, converting it into heat~\cite{landauer1961irreversibility,bennett2003notes}. A simple resource-theoretical model of an erasure process consists of a unital channel acting on a system $S$ whose state is to be erased (i.e.~mapped to a fixed pure state $\ket{\psi}$) together with an \emph{information battery} $B$ that acts as a source of purity. In the simplest setting, $B$ is an $n$-qubit system, with each qubit being either in a pure state $\ket{0}$ or maximally mixed. We define the \emph{finite-size work cost} (in units of $k_B\ln(2)$) of erasing an initial state $\rho$ as the size of the smallest information battery that allows for an erasure of $\rho$, i.e.~the smallest integer $n$ such that 
    \begin{equation}
        \rho \otimes \proj{0}^{\otimes n}  \succ_{\id}\proj{\psi} \otimes (\one_2/2)^{\otimes n}.
    \end{equation}
    The usual formulation of Landauer's bound, $n \geq S(\rho)$, is then a simple consequence of the monotonicity of the entropy. However, applying Result~\ref{res:bound_entropy_production} yields, 
    \begin{equation} \label{eq:landauer_correction}
        n  \geq S(\rho) + \frac{V(\rho)}{2\sqrt{M(\rho)}},
    \end{equation}
	which provides corrections to Landauer's bound that increase with the initial variance of $\rho$. 
    Indeed, we note that Eq.~\eqref{eq:landauer_correction} remains true even when the battery size is not constrained, i.e.~one can allow the initial and final battery states to contain arbitrarily large ``reservoirs" of pure and maximally mixed qubits, that is, states of the form $ \proj{0}^{\otimes \lambda_1}\otimes (\one_2/2)^{\otimes \lambda_2}  $ and $  \proj{0}^{\otimes (\lambda_1-n)}\otimes (\one_2/2)^{\otimes (\lambda_2+n)} $ for arbitrarily large $ \lambda_1,\lambda_2 $. At the same time, the correction term in Eq.~\eqref{eq:landauer_correction} can easily be made to vanish in the presence of a bystander system whose state is returned unchanged and uncorrelated from $S$ and $B$. Such systems are technically known as trumping catalysts (see Box~\ref{box:catalysis_and_bounds}) and hence the above bound is not robust to this simple extension of the setting.
\end{frameenv}

\subsection{Bounds on marginal entropy production}
\label{sub:bipartite_tradeoff_and_local_monotonicity}

In the previous subsection, it was shown that a decrease in the variance lower-bounds entropy production under unital channels. However, not all quantum channels are unital channels and hence it is natural to ask whether a similar result exists for a more general class of channels. It is clear that Result~\ref{res:bound_entropy_production} cannot be generalized for \emph{all} channels: for instance, the channel that maps every input state to a pure state reduces both entropy and variance. However, using the properties of the variance and the previous results, we can formulate an analogous lower bound for arbitrary quantum channels, by considering their effect on the \emph{environment}. 
To state those results, we first note the well-known fact that every quantum channel on a quantum system $S$ can be understood as the local effect of a unital channel acting on $S$ together with an environment $E$. More formally, any channel $\mathcal{E}$ from $S$ to itself can be written as
\begin{equation}\label{eq:channeldilation}
    \mathcal{E}(\cdot) = \tr_E [{\mc U}(\cdot\otimes \rho_E)]
\end{equation}
for some initial state $\rho_E$ of the environment and unital channel $\mc U$ on the joint system $SE$. We call the pair $(\mc U, \rho_E)$ a \emph{dilation} of $\mc E$. By the Stinespring dilation theorem, one can always choose $\rho_E$ to be pure and $\mc U$ to be a unitary channel for sufficiently large environment dimension $d_E$,
but here, in the context of the resource theory of purity, we use the above more general representation.

We next show that one can extend Result~\ref{res:bound_entropy_production} to \emph{local} changes of variance and entropy. Let $(\mc U, \rho_E)$ be a dilation of a given quantum channel $\mc E$. For an initial $\rho_S$ on $S$, let $\rho'_S = \mc E(\rho_S)$ denote the final state on $S$ and let $\rho'_E = \tr_S(\mc U(\rho_S \otimes \rho_E))$ denote the final state on $E$. Moreover, denote as
\begin{align}
    \Delta S_S & = S(\rho_S) - S(\rho'_S), \\
    \Delta V_S & = V(\rho_S) - V(\rho'_S)  
\end{align} 
the changes of entropy and variance on $S$ respectively, and similarly for the environment. Finally, let 
\begin{equation} \label{eq:def_mutual_information}
    I_{S:E} = S(\mc U(\rho_S \otimes \rho_E) \| \rho'_S \otimes \rho'_E)
\end{equation}
denote the mutual information between $S$ and $E$ after the application of $\mc U$. We then have the following:
\begin{result}[Lower bound on \emph{marginal} entropy production] \label{res:lower_bound_bipartite_unital}
    Given a quantum channel $\mc E$ from $S$ to itself, let $(\mc U, \rho_E)$ be any dilation of $\mc E$  {where $ \mc U $ is a unital map,} and denote $d_{SE} = d_S \cdot d_E$. Then, for any $\rho_S \in \mc D(S)$,
    \begin{align} 
        \!\!\!\!\!\!\!\!          -\Delta S_S - \Delta S_E & \geq \frac{\Delta V_S + \Delta V_E - K' \log(d_{SE})^2 f(I_{S:E})}{2 \sqrt{M(\rho_S \otimes \rho_E)}} \label{eq:bound_var_local_entropy} 
    \end{align}
    where $K'$ is a constant independent of $d$ or $\rho$ and $f(x) := \max\{\sqrt[4]{x}, x^2\}$.
\end{result}
This result follows straightforwardly from combining Result~\ref{res:bound_entropy_production} with Property $4$ of the variance (or more precisely, Lemma~\ref{lem:subadditivity_variance}) together with the subadditivity of the von Neumann entropy.
It is particularly interesting from a resource-theoretic point of view, since in such theories, the environment $E$ often explicitly models a particular kind of physical system, such as a thermal bath, a clock, a battery, or a catalyst. 
For instance, we may then consider the setting of Landauer erasure described in Box~\ref{appl:erasure}, with the additional requirement that $I_{S:E} = 0$. In Box~\ref{box:catalysis_and_bounds}, we also apply Result~\ref{res:lower_bound_bipartite_unital} to gain insight into catalytic processes, in particular, to derive bounds on the dimension of the catalyst required for certain processes.

\begin{frameenv}{Bound on catalyst dimension for state transitions} 
    \label{box:catalysis_and_bounds}                            
    It is well-known that the set of possible state transitions in a resource theory can be enlarged with the help of catalysts, that is, auxiliary systems whose local state remains unchanged in a process. In terms of the notation established in the main text and given two quantum states $\rho, \rho' \in \mc D(S)$ , we write $\rho \succ_C \rho'$ if there exists a quantum channel $\mc E$ with dilation $(\mc U, \rho_E)$ such that $\rho'_E$ = $\rho_E$, that is, if the local state of the environment remains unchanged. Moreover, we write $\rho \succ_{T} \rho'$ if the dilation can be chosen such that $I_{S:E} = 0$, where the catalyst not only remains locally unchanged, but is also returned uncorrelated from $S$. This relation is known as \emph{trumping}~\cite{turgut2007necessary,klimesh2007inequalities,PhysRevLett.83.3566}. Clearly, 
    \begin{equation}
        \rho \succ_{\id}\rho' \Rightarrow \rho \succ_T \rho' \Rightarrow \rho \succ_C \rho',
    \end{equation}
    while the converse relations in general do not hold. As such, catalysts enable previously impossible state transitions. 
                            
    Indeed, recently it was shown that if $\rho$ and $\rho'$ are two full-rank states, then $\rho \succ_C \rho'$ is equivalent to $S(\rho') > S(\rho)$~\cite{Muller2018}, a result that has found applications in the context of fluctuation theorems in (quantum) thermodynamics~\cite{Boes2019} and can be further strengthened to the case that $\mc U$ is unitary~\cite{Boes2018, wilming2020entropy}. What these results are silent about, however, is the required size of the catalyst. Here, we apply Result~\ref{res:lower_bound_bipartite_unital} to show that transitions between states with similar entropy that decrease the variance can only be realized by means of a catalyst with very large dimension. In particular, consider any state transition $\rho \succ_C \rho'$ between full-rank states such that $0 \leq S(\rho') - S(\rho) \leq \delta \leq 1$. Then, Result~\ref{res:lower_bound_bipartite_unital} implies that 
    \begin{equation}
        \Delta V_S \leq \tilde{K} \log(d_{SE})^2 \sqrt[4]{\delta}.
    \end{equation}
    This follows from the fact that $\Delta V_E = \Delta S_E = 0$, the monotonicity of $f$ and the logarithm, as well as $I_{S:E} \leq S(\rho') - S(\rho)$. This shows that, for any fixed $\Delta V_S$ and fixed system dimension $d_S$, $d_E$ has to grow as $d_E\geq O(\exp(\delta^{-1/8}))$ for the above equation to be satisfied. For state transitions where the entropy change is small, reducing the variance is therefore possible only at the expense of using a large catalyst.
\end{frameenv}

\subsection{Local monotonicity and entropy} 
\label{ssub:local_monotonicity}
Result~\ref{res:lower_bound_bipartite_unital} is non-trivial only when the RHS of Eq.~\eqref{eq:bound_var_local_entropy} is positive (i.e.~when there is significant decrease of marginal variances compared to the mutual information). This is because the LHS is always non-negative --- a property we call \emph{local monotonicity} \emph{with respect to maximally mixed states (and unital channels)}. The local monotonicity of entropy follows straightforwardly from the fact that it is Schur concave, additive and sub-additive. Our last result is to show that, conversely, this property is essentially \emph{unique} to the von Neumann entropy, namely, it singles out the latter from all continuous functions on quantum states. 
                    
To state our result, let us first define local monotonicity more formally. Consider a function $f$ on quantum states on finite-dimensional Hilbert-spaces. Let $\id_1$ and $\id_2$ be maximally mixed-states on systems $S_1$ and $S_2$, and let $\mc U$ be a unital channel, namely $\mc U(\id_1\otimes\id_2)=\id_1\otimes\id_2$. We say that $f$ is \emph{locally monotonic with respect to maximally mixed states} if for any such $\id_i$, $\mc U$ and two states $\rho_i\in \mc D(S_i)$, we have
\begin{align}
    f(\rho_1) + f(\rho_2) \leq f(\rho_1') + f(\rho_2'), 
\end{align}
where $\rho_1'=\tr_{2}{\mc U}(\rho_1\otimes \rho_2)$ and similarly for $\rho_2'$.                     
We then have the following result:
                    
\begin{result}[Uniqueness of von~Neumann entropy] \label{res:entropy_single}
    Let $f$ be a continuous function that is locally monotonic with respect to maximally mixed states. Then
    \begin{equation}
        f(\rho) = a S(\rho) + b_d, 
    \end{equation}
    where $S$ is the von~Neumann entropy, $d$ is the Hilbert-space dimension of $\rho$ and $b_{d}$ depends only on $d$ but not otherwise on $\rho$. 
    It is sufficient for this result to restrict the set of channels to unitary channels.
\end{result}
The proof of the result can be found in Appendix~\ref{app:single_out_entropy}.
While there exist many axiomatic characterizations of the entropy, we consider the above interesting because its only axiom (apart from continuity) --- local monotonicity --- is directly motivated by physical, as opposed to mathematical, considerations. This is because physics is often concerned with the possible changes of local quantities in the course of physical processes. 
As an example, in Box~\ref{box:local_monotonicity}, we apply Result~\ref{res:entropy_single} to derive a version of Boltzmann's H theorem. 

Finally, going back to Result~\ref{res:lower_bound_bipartite_unital}, we see that just like how Result~\ref{res:bound_entropy_production} provides a strengthening to the monotonicity of the entropy, Result~\ref{res:lower_bound_bipartite_unital} provides a strengthening to the local monotonicity of the entropy.

\begin{frameenv}{A version of Boltzmann's H theorem}\label{box:local_monotonicity}

    Here we note that one can derive a version of Boltzmann's \emph{H theorem} as a simple corollary of Result~\ref{res:entropy_single}: Consider a ``gas'' of $N$ independent quantum systems initially in state $\rho^{(0)} = \otimes_{i=1}^N \rho_{i}$. At any point in time $t$, two of these systems, call them $i$ and $j$, first undergo a joint evolution, described by a (possibly random) unitary channel $\mc{U}_t$. We further assume that, following this interaction, any correlations between these two particles vanish (This is the infamous ``Sto\ss zahlansatz''). A single iteration of this process then yields the chain of states
    \begin{equation}
        \rho^{(t)}_{ij} = \rho_{i} \otimes \rho_{j} \to \rho'_{ij} = \mc{U}_t(\rho^{(t)}_{ij}) \to \rho^{(t+1)}_{ij} = \rho'_{i} \otimes \rho'_{j},
    \end{equation}
    where $\rho_{ij} = \tr_{(ij)^c}[\rho]$. All other systems in the gas remain unchanged during this process. We are now interested in finding a continuous, real-valued function $f$ such that 
    \begin{equation}
        f(\rho^{(t)}) \leq f(\rho^{(t+1)})
    \end{equation}
    for all $t$ and all possible $\mc{U}_t$. Then, the above result implies that $f$ exists and is given by the von~Neumann entropy.
    \vspace{0.3cm}
\end{frameenv}

\section{Main results for generic quantum channels} 
\label{sec:relative_results}
In the previous section, we have provided an overview of our main results for the special case of the resource theory of purity. We now turn to an exposition of our results in their full generality.
All previously mentioned results are special cases of the general results presented here. Since we discussed the interpretation of the results already in the last section, we now focus on the technical and formal presentation. Some of technical proofs are nevertheless delegated to the appendices.

\subsection{Notation and main concepts} 
\label{sub:notation_and_main_concepts}

 {Recall the pre-order  {$\succ$ defined at the beginning of the last section for pairs of states $(\rho, \sigma)$.}
A well-known example of this ordering that goes beyond majorization is $\beta$-majorization in quantum thermodynamics, where $\sigma = \sigma'$ is the thermal state of the system at inverse temperature $\beta$~\cite{horodecki2013fundamental}.
Throughout the following, we focus on the \emph{quasi-classical} setting, in which  {the initial and final pairs commute with one another, i.e. $[\rho, \sigma] = [\rho', \sigma'] = 0$. 
This significantly simplifies the setting mathematically, while still allowing some genuine quantum features, for example when $ \sigma = \sigma' $ but $[\rho, \rho'] \neq 0$.}

\subsubsection{Lorenz curves}

A key tool in proving our sufficiency results is a well-known connection between  {the pre-order $  {\succ} $} 
and the \emph{Lorenz curve}. For self-consistency, we present this connection and all relevant constructions in the notation of this paper. %
In particular, let $\rho$ and $\sigma$ be two commuting positive-semidefinite operators on the same $d$-dimensional Hilbert space $\mc H$, and denote by $\{\ket{i}\}_{i=1}^d$ an orthonormal basis of $\mc H$ that simultaneously diagonalizes both $ \rho $ and $ \sigma $, i.e. $\sigma = \sum_i s_i \proj{i}$ and $\rho = \sum_{i} p_i \proj{i}$. Furthermore, let us assume w.l.o.g. that the basis $\{\ket{i}\}_{i=1}^d$ \emph{orders $\rho$ relative to $\sigma$}, namely
\begin{equation}\label{eq:sigma_ordering}
    \frac{p_i}{s_i} \geq \frac{p_{i+1}}{s_{i+1}}
\end{equation}
for any $i=1,\ldots,d$. 
Note that in this ordering neither the $(p_i)_i$ nor $(s_i)_i$ are necessarily ordered. Given the notations introduced, we can now introduce the \emph{Lorenz curve}.
\begin{definition}[Lorenz curves]\label{def:Lorenz_curve}
    Given  {two commuting quantum states $\rho$ and $\sigma$}, the Lorenz curve $ \mathcal{L}_{\rho|\sigma}(x):[0,1]\rightarrow[0,1] $ is given by the piecewise linear curve that connects the points
    \begin{equation}
        \left\lbrace \sum_{i=1}^k s_i, \sum_{i=1}^k p_i \right\rbrace_{k=1}^d.
    \end{equation}
     {If $\rho$ and $\sigma$ do not commute},
    the Lorenz curve is taken to be $ \mathcal{L}_{\rho|\sigma} =  \mathcal{L}_{\mathcal{W}(\rho)|\sigma} $, where $ \mathcal{W}(\rho) $ is the state pinched to the eigenbasis of $ \sigma $, i.~e.~, $\mc W(\rho) = \sum_i P_i \rho P_i$, with $P_i$ the projectors onto the eigenspaces of $\sigma$. 
\end{definition} 
Due to the way we have ordered the eigenvalues according to Eq.~\eqref{eq:sigma_ordering}, the Lorenz curve is by definition always concave. 
The following now provides a simple  {and well-known} equivalence relation between Lorenz curves and $\sigma$-majorization.
\begin{theorem}\label{thm:majorization_vs_Lorenz_curves}
     {Given two pairs of commuting states $(\rho, \sigma)$ and $(\rho', \sigma')$,}
    the following are equivalent:
    \begin{enumerate}[nolistsep,noitemsep]
        \item For the entire range of $x\in[0,1]$,
              \begin{equation}\label{eq:Lrhosigma}
                \mc{L}_{\rho|\sigma}(x) \geq  {\mc{L}_{\rho'|\sigma'}(x)}.
              \end{equation} 
        \item  {$(\rho, \sigma) \succ (\rho', \sigma')$.}
    \end{enumerate}
\end{theorem}
Theorem \ref{thm:majorization_vs_Lorenz_curves} is a condensed version of a more extended statement found  {as Theorem \ref{thm:Renes} in Appendix~\ref{app:auxiliary_lemmata}, also see~\cite{ruch1978mixing,RUCH1980222,blackwell1953,Marshall2011}}. 

\subsubsection{Flat and steep approximations relative to $\sigma$} %
\label{sub:flattest_and_steepest_approximation_relative_to_sigma_}
We are now in a position to define the following approximations, known as flat and steep approximations of a state $\rho$ relative to $\sigma$, denoted as $\rho_{\rm fl}^\epsilon$ and $\rho_{\rm st}^\epsilon$ respectively, which will play an important role for the derivation of our results. These states were initially defined in Ref.~\cite{VanderMeer2017} for the special case of thermal reference states. 
Although the following Definitions~\ref{def:flattest_approx} and~\ref{def:steepest_approx} seem technical, they have the essential appealing property that for any state $\rho$ and any $1>\epsilon >0$, we have~\footnote{In fact, for any state $\hat\rho$ such that $ D(\rho,\hat\rho)\leq \epsilon $, we have that $ \hat\rho \succ_\sigma \rho_{\rm fl}^\epsilon$, and therefore $\rho_{\mathrm{fl}}^\epsilon$ is also known as the flattest state. 
    The analogous statement is however not true for $\rho_{\mathrm{st}}^\epsilon$, as~\cite{VanderMeer2017} shows that there is no unique $ \epsilon $-steepest state in general.}
\begin{align}
    \rho_{\rm st}^\epsilon \succ_\sigma \rho \succ_\sigma \rho_{\rm fl}^\epsilon. 
\end{align}
The states are constructed as follows. 

\begin{definition}[Flat approximation relative to $\sigma$] \label{def:flattest_approx}
    Let $\sigma, \rho$ be commuting quantum states on a $d$-dimensional Hilbert space $\mc H$ and $\{\ket{i}\}_{i=1}^d$ a common eigenbasis of the two states that orders $\rho$ relative to $\sigma$, yielding $\sigma = \sum_i s_i \proj{i}$ and $\rho = \sum_{i} p_i \proj{i}$. 
    For any $0 \leq \epsilon \leq 1$, the \emph{$\epsilon$-flattest approximation relative to $\sigma$} is the state $\rho_{\mathrm{fl}}^\epsilon = \sum_i \bar{p}_i \proj{i}$, where the $\bar{p}_i$ are defined as follows: If $D(\rho, \sigma) < \epsilon$, set $\bar{p}_i = s_i$. Otherwise, define $M \in \lbrace 1,2,\cdots,d-1\rbrace$ as the smallest integer such that 
    \begin{equation}\label{eq:int_M}
        \epsilon \leq \sum_{i=1}^M p_i - \frac{p_{M+1}}{s_{M+1}} \sum_{i=1}^M s_i
    \end{equation}
    and let $N \in \{2, \dots, d\}$ be the largest integer such that 
    \begin{equation}\label{eq:int_N}
        \epsilon \leq \frac{p_{N-1}}{s_{N-1}} \sum_{i=N}^d s_i - \sum_{i=N}^d p_i.
    \end{equation} 
    These integers always exist when $\epsilon \leq D(\rho, \sigma)$ and moreover satisfy $M \leq N$ (\cite{VanderMeer2017}, App. D, Lemma 6). Using these definitions, finally set 
    \begin{align}
        \bar{p}_i                                                   & =                    
        \begin{cases}
        s_i \frac{(\sum_{j=1}^M p_j) - \epsilon}{\sum_{j=1}^M s_j}, & \text{ if } i \leq M \\
        s_i \frac{(\sum_{j=N}^d p_j) + \epsilon}{\sum_{j=N}^d s_j}, & \text{ if } i \geq N \\
        p_i                                                         & \text{ otherwise.}   
        \end{cases}
    \end{align}
\end{definition}
                
\begin{definition}[Steep approximation relative to $\sigma$] \label{def:steepest_approx}
     {Let $\sigma, \rho$ be commuting quantum states on a $d$-dimensional Hilbert space $\mc H$ and $\{\ket{i}\}_{i=1}^d$ a common eigenbasis of the two states that orders $\rho$ relative to $\sigma$, yielding $\sigma = \sum_i s_i \proj{i}$ and $\rho = \sum_{i}^d p_i \proj{i}$.}
    Then, for $0 \leq \epsilon \leq 1$, the \emph{$\epsilon$-steep approximation relative to $\sigma$} is the state $\rho_{\mathrm{st}}^\epsilon = \sum_i \hat{p}_i \proj{i}$, such that if $ \epsilon\leq 1-p_1 $, 
    \begin{align}
        \hat{p}_i          & =                     
        \begin{cases}
        p_1 + \epsilon,    & \text{ if } i =1      \\
        p_i ,              & \text{ if } 1 < i < R \\
        p_i -(\epsilon-r), & \text{ if } i = R     \\                                                     
        0                  & \text{ otherwise,}    
        \end{cases}
    \end{align}
    where $R \in \{2, \dots, d\}$ is the largest index such that $\sum_{i=R}^d p_i \geq \epsilon$ and by definition of $ R $, we have $r = \sum_{i=R+1}^d p_i \leq \epsilon$. On the other hand, if $\epsilon > 1 - p_1$, define 
    \begin{align}
        \hat{p}_i & =                  
        \begin{cases}
        1,        & \text{ if } i = 1  \\
        0         & \text{ otherwise.} 
        \end{cases}
    \end{align}
\end{definition}

\begin{figure}
    \includegraphics[trim=2.65cm 9cm 3cm 9.5cm,clip, width=0.47\textwidth]{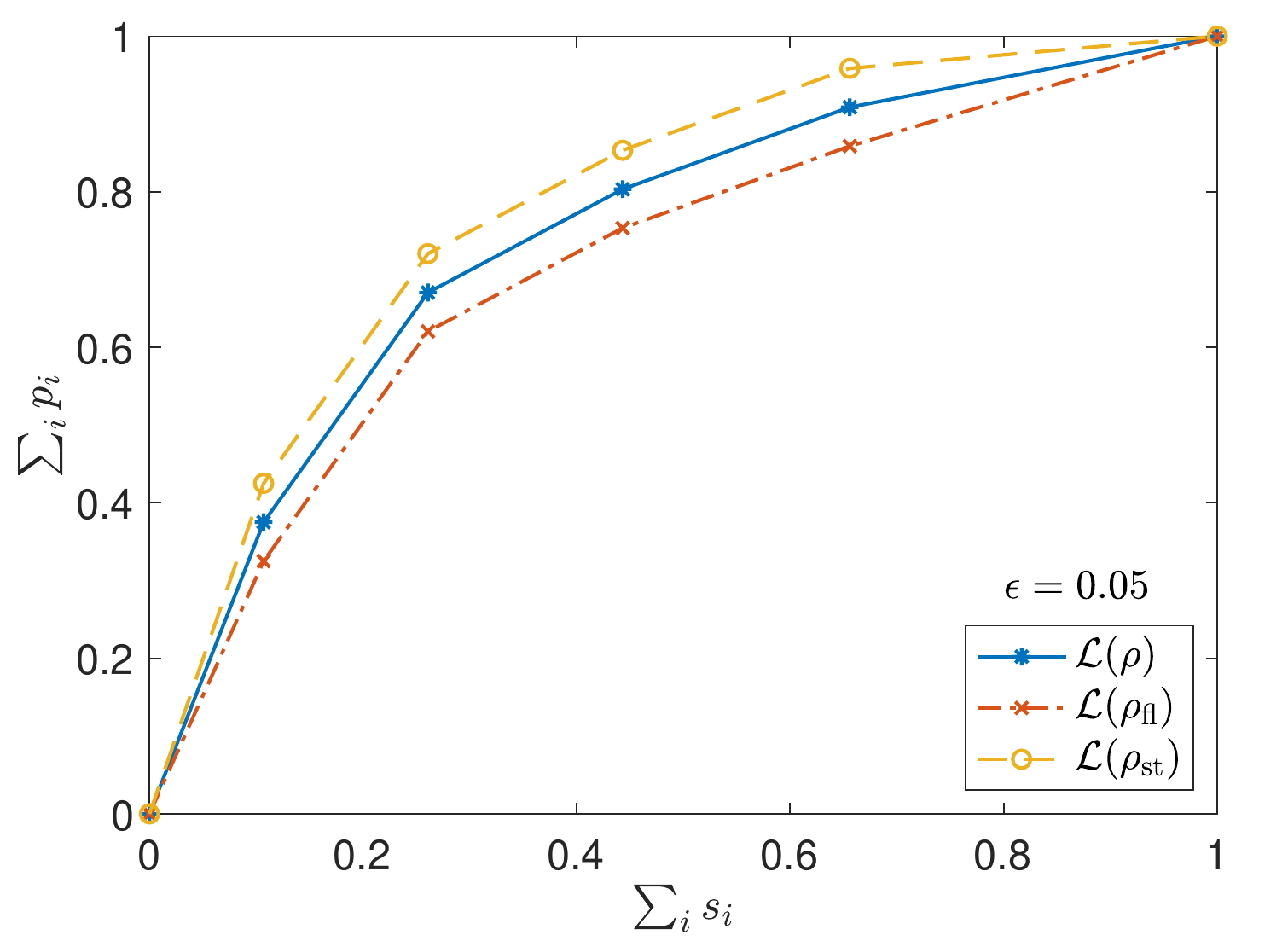}
    \caption{An example illustrating the Lorenz curves $\mathcal{L}(\rho), \mathcal{L}(\rho_{\rm st}^\epsilon), $ and $  \mathcal{L}(\rho_{\rm fl}^\epsilon) $, for some $ \epsilon = 0.05 $, and a state $\rho$ of rank 5. The construction according to Def. \ref{def:flattest_approx} and \ref{def:steepest_approx} are so that $\mathcal{L}(\rho_{\rm st}^\epsilon)\geq  \mathcal{L}(\rho) \geq  \mathcal{L}(\rho_{\rm fl}^\epsilon) $ for any value $ \epsilon \in [0,1]$ and any state $ \rho $.} 
        \label{fig:steepflat}
\end{figure}
Figure~\ref{fig:steepflat} illustrates an example of steep and flat approximations. We now state a key technical lemma, that provides the properties of Lorenz curves for $\epsilon$-steep and $\epsilon$-flat approximations for any  {pair $(\rho, \sigma)$}. 

\begin{lemma}\label{lem:variance_in_steepflat_properties}
     {Let $(\rho, \sigma)$ be two commuting quantum states, with $\sigma$ full-rank} and let $1>\epsilon > 0$. Then, for $ x\in [0,1] $,
    \begin{align}\label{Lsteepflat}
        \mathcal{L}_{\rho_{\mathrm{st}}^\epsilon|\sigma} (x) & \geq \ell_{r_{\rm st}}(x),~~~ 
        r_{\rm st} = 2^{S(\rho \| \sigma) - f_\sigma(\rho,\epsilon)},\\
        \mathcal{L}_{\rho_{\mathrm{fl}}^\epsilon|\sigma} (x) & \leq \ell_{r_{\rm fl}}(x),~~~ 
        r_{\rm fl} = 2^{S(\rho \| \sigma) +  f_\sigma(\rho,\epsilon)},
    \end{align}
    where $ f_\sigma(\rho,\epsilon) := \sqrt{V(\rho\| \sigma) \left(\epsilon^{-1}-1\right) }$ and $ l_c(x) = {\rm min} (c\cdot x, 1)$. 
\end{lemma}

This lemma is proven in Appendix~\ref{app:lemma_variance_steepflat}. Intuitively, it shows that the steep (flat) approximations allow us to obtain a state close to $\rho$ in trace distance, with its Lorenz curve being lower (upper) bounded by straight lines $l_c$ with gradients $c=r_{\rm st}(r_{\rm fl})$ governed by both the relative entropy and its variance. These simple bounds on the Lorenz curves of $\rho_{\mathrm{st}}$ and $\rho_{\mathrm{fl}}$ are crucial for our derivation of Theorem~\ref{res:suff_transition_relative} and Theorem~\ref{thm:bound_divergences_by_variance}, which are the general statements for Results~\ref{res:suff_transition_unital} and~\ref{res:bound_min_max_by_variance} and are obtained as a direct consequence of this Lemma~\ref{lem:variance_in_steepflat_properties}.

\subsection{Sufficient criteria for state transitions under $\sigma$-majorization} 
\label{sub:sufficiency_conditions}

Using Lemma~\ref{lem:variance_in_steepflat_properties}, we can now derive sufficiency conditions for approximate state transitions between commuting pairs of quantum states.

\begin{theorem} \label{res:suff_transition_relative}
     {Let $(\rho, \sigma)$ and $(\rho', \sigma')$ be two pairs of commuting states with $\sigma,\sigma'$ both full-rank.}
    For $ 1 > \epsilon > 0$, let $f_\sigma(\rho,\epsilon) := \sqrt{V(\rho\| \sigma) \left(2\epsilon^{-1}-1\right)}$. If
    \begin{equation}
	    S(\rho\| \sigma) - f_\sigma(\rho,\epsilon) ~ {\geq}~ S(\rho'\|  {\sigma'}) + f_{\sigma'}(\rho',\epsilon), 
    \end{equation}
	then  {$(\rho, \sigma) \succ_\epsilon (\rho', \sigma')$}.
\end{theorem}
\begin{proof}
    Let $ \bar\epsilon = \epsilon/2 $, and $r_{\rm st} = 2^{S(\rho\| \sigma) - f_\sigma(\rho,\bar\epsilon)}$, and  {$ r_{\rm fl}' = 2^{S(\rho'\| \sigma') + f_{\sigma'}(\rho',\bar\epsilon)} $}. Note that if the above condition holds, then in the whole range of $x\in [0,1]$, we have that $l_{r_{\rm st}} \geq l_{r_{\rm fl}'}$. By Lemma~\ref{lem:variance_in_steepflat_properties}, we then have that in that range
    \begin{equation}\label{eq:somesteepflatequation}
        \mathcal{L}_{\rho_{\mathrm{st}}^{\bar\epsilon}|\sigma} (x)\geq \ell_{r_{\rm st}}(x) \geq l_{r_{\rm fl}'}(x) \geq \mathcal{L}_{\rho_{\mathrm{fl}}'^{\bar\epsilon}| {\sigma'}} (x),
    \end{equation}
    which by Theorem \ref{thm:majorization_vs_Lorenz_curves} implies that there exists a channel  {$\mc{E}$ such that $\mc E(\sigma) = \sigma'$ and} $\mc{E}(\rho_{\mathrm{st}}^{\bar\epsilon}) = \rho_{\mathrm{fl}}'^{\bar\epsilon}$. Applying the same channel to $\rho$ yields a state $\mc{E}(\rho) = \hat\rho'$ such that $D(\hat\rho',\rho')\leq \epsilon $, since
    \begin{align}\label{triangleDflat}
        D(\hat\rho',\rho') & \leq D(\hat\rho',\rho_{\mathrm{fl}}'^{\bar\epsilon}) + D(\rho_{\mathrm{fl}}'^{\bar\epsilon},\rho') \\
                           & \leq D(\rho,\rho_{\mathrm{st}}^{\bar\epsilon}) + \bar\epsilon \leq 2\bar\epsilon = \epsilon.       
    \end{align}
\end{proof}

Result~\ref{res:suff_transition_unital} in Section~\ref{sec:overview_of_main_results} follows as a special case for $\sigma = \id$. 
As mentioned earlier, in Appendix~\ref{app:second-order asymptotics}, we apply Theorem~\ref{res:suff_transition_relative} to derive sufficient conditions for i.i.d.~state transitions with large but finite number of states, recovering a previously observed resonance condition, where second-order corrections can vanish even for non-zero variances of the initial and final states $ \rho,\rho' $.

\subsection{Relation to smoothed min- and max-relative entropies}
\label{ssub:min_and_max_relativ_entropies}

Two quantities that are useful in describing single-shot processes are the min- and max-relative entropy. Given a positive semidefinite operator $\sigma \geq 0$ and a quantum state $\rho \in \mc D(S)$, let $\pi_\rho$ denote the projector onto the support of $\rho$. Moreover, for two operators $A$ and $B$, we write $A \geq B$ to mean that the operator $A - B$ is positive semidefinite. In terms of this notation, if $\mathrm{supp}(\rho) \subseteq \mathrm{supp}(\sigma)$, then we have the following definitions~\cite{Datta_2009}:
\begin{align}
    S_{\min}(\rho \| \sigma)  & := -\log \tr(\pi_\rho \sigma),                       \\
    S_{\max}(\rho \| \sigma) & := \log \min \{ \lambda: \rho \leq \lambda \sigma\}. 
\end{align}
The \emph{smoothed} variants are further defined as
\begin{align}
    S_{\min}^\epsilon(\rho \| \sigma) & := \displaystyle\max_{\tilde{\rho}\in\mathcal{B}^\varepsilon (\rho)} S_{\min}(\tilde{\rho}||\sigma), \\
    S_{\max}^\epsilon(\rho \| \sigma) & := \displaystyle\min_{\tilde{\rho}\in\mathcal{B}^\varepsilon (\rho)} S_{\max}(\tilde{\rho}||\sigma),      
\end{align}
where the optimizations are over the set of all quantum states $ \varepsilon $-close in terms of trace distance to $ \rho $, denoted as $ \mathcal{B}^\varepsilon (\rho)$. Finally, define as $S^\epsilon_{\max}(\rho) := \log(d) - S^\epsilon_{\min}(\rho \| \id)$ and $S^\epsilon_{\min}(\rho) := \log(d) - S^{\epsilon}_{\max}(\rho \| \id)$ the min- and max-entropies utilized in Result~\ref{res:bound_min_max_by_variance}.

We now present the generalization of Result~\ref{res:bound_min_max_by_variance} for the smoothed min- and max-relative entropies, which is easily proven by making use of Lemma~\ref{lem:variance_in_steepflat_properties}.

\begin{theorem}[Bounds on smoothed R{\'e}nyi divergences] \label{thm:bound_divergences_by_variance}
     {Given $\rho, \sigma \in \mc D(S)$, let} $1 > \epsilon > 0$. Then,
    \begin{align} \label{eq:bound_min_div_by_variance_relative}
        S_{\max}^{\epsilon}(\rho\|\sigma) - S(\rho\|\sigma) & \leq f_\sigma(\rho,\epsilon), \nonumber \\
        S(\rho\|\sigma) - S_{\min}^{\epsilon}(\rho\|\sigma) & \leq f_\sigma(\rho,\epsilon),\nonumber  
    \end{align}
    where $ f_\sigma(\rho,\epsilon) := \sqrt{V(\rho\| \sigma) \left(\epsilon^{-1}-1\right)}$.
\end{theorem}
\begin{proof}
    We know from Ref.~\cite{VanderMeer2017} that $S_{\max}^\epsilon(\rho\|\sigma) = S_{\max}(\rho_{\mathrm{fl}}^\epsilon\|\sigma)$. Therefore, 
    \begin{align}
        S_{\max}^{\epsilon}(\rho\|\sigma) - S(\rho\|\sigma) & = S_{\max}(\rho_{\mathrm{fl}}^\epsilon\|\sigma)-S(\rho\|\sigma) \leq f_\sigma(\rho,\epsilon). 
    \end{align}
    The last inequality follows from Lemma~\ref{lem:variance_in_steepflat_properties}, which implies
    \begin{equation}\label{someSinftyequation}
        S_{\max}(\rho_{\mathrm{fl}}^\epsilon\|\sigma) \leq \log r_{\rm fl}  = S(\rho \| \sigma) + f_\sigma(\rho, \epsilon),
    \end{equation}
    since $S_{\max}(\rho\|\sigma)$ is simply the logarithm of the gradient of the Lorenz curve $\mc{L}_{\rho|\sigma}$ at the origin. 
    On the other hand, 
    \begin{align}\label{eq:maxSinfty_vs_steep}
        S_{\min}^\epsilon (\rho\|\sigma) & = \max_{\tilde{\rho}\in\mathcal{B}^\varepsilon (\rho)} S_{\min}(\tilde\rho\|\sigma)\geq S_{\min}(\rho_{\rm st}^\epsilon\|\sigma). 
    \end{align}
    Now, let $\pi_\rho$ denote the projector onto the support of $\rho$. By definition of the Lorenz curve we have that 
    \begin{equation}
        \tr(\pi_\rho \sigma) = \min \lbrace x | x\in [0,1], {\mc L}_{\rho |\sigma} (x) =1  \rbrace.
    \end{equation}
    Using this fact and Lemma~\ref{lem:variance_in_steepflat_properties}, in particular the definition of $l_c(x)$, we then find
    \begin{equation}\label{somesum}
        \tr(\pi_{\rho_\mathrm{st}^\epsilon} \sigma) \leq r_{\rm st}^{-1}.
    \end{equation}
    Combining this with the definition of the smooth min-relative entropy then yields 
    \begin{equation}\label{eq:apply_lemvar_bound_S0steep}
        S_{\min}^\epsilon (\rho\|\sigma) \geq S_{\min} (\rho_{\mathrm{st}}^\epsilon\|\sigma) \geq \log r_{\rm st} = S(\rho\|\sigma) -f_\sigma(\rho,\epsilon).
    \end{equation}
    Finally, combining Eqns.~\eqref{eq:maxSinfty_vs_steep} and \eqref{eq:apply_lemvar_bound_S0steep} yields the second claim of the theorem.
\end{proof}

\subsection{Uniform continuity and correction to subadditivity} %
\label{sub:uniform_continuity_and_correction_to_subadditivity}
We here show that the variance of relative surprisal is uniformly continuous and bound its violation of subadditivity. Both of these properties of the relative variance are key tools to derive our main results, however we believe that they are of independent interest and use.
\begin{lemma}[Uniform continuity of the relative variance] \label{lem:uniform contintuity relative variance}
     {Let $\sigma$ be a positive-definite operator with $0 < \sigma \leq \id $ and smallest eigenvalue $s_{\min}$
on a $d$-dimensional Hilbert space with $d \geq 2$, and $\rho, \rho'$ be two states that both commute with $\sigma$}. Then, 
    \begin{align}
        |V(\rho\|\sigma) - V(\rho'\|\sigma)| \leq 2K\sqrt{ {D(\rho, \rho')}}, 
    \end{align}
    where $K = 8 \log^2(d) + \log(d) + 2 \log^2(s_{\min}) - 4 \ln(2) \log(s_{\min}) + 15$.
\end{lemma} 
This lemma is proven in Appendix~\ref{app:proof_uniform_continuity}.
\begin{lemma}[Correction to sub-additivity of relative variance.] \label{lem:subadditivity_variance}
     {Let $\rho, \sigma$ be two commuting quantum states on a $d$-dimensional, bipartite system, with $d \geq 2$ and $\sigma$ full-rank. If } $\sigma = \sigma_1\otimes\sigma_2 $ is a product state with smallest eigenvalue $s_{\min}$, then
    \begin{align}
        V(\rho\|\sigma) \leq V(\rho_1\|\sigma_1) + V(\rho_2\|\sigma_2) +  K'\cdot f(I_{1:2}), 
    \end{align}
    where $K' = \sqrt{2\ln(2)}\left(12 + \log(s_{\min})^2 + 8 \log^2(d)\right)$, $f(x) = \max\{\sqrt[4]{x}, \sqrt{x}\}$ and $ I_{1:2} $ denotes the mutual information between the two partitions of $ \rho $.
\end{lemma}
This lemma is proven in Appendix~\ref{app:proof_of_correction_to_subadditivity_of_relative_variance_}.

\subsection{A new monotone and relative entropy production}
\label{sub:necessary_conditions_for_sigma_majorization}

We now turn to the presentation and derivation of the results that generalize Results~\ref{res:bound_entropy_production} and~\ref{res:lower_bound_bipartite_unital}. We begin by noting that  { the relative entropy $ S(\rho\|\sigma) $ is a non-increasing resource monotone with respect to the ordering $\succ$}, generalizing the Schur concavity of the von Neumann entropy. 
We then have the following generalization of Lemma~\ref{res:monotone}, which we prove in App.~\ref{app:proof_schur_concavity}. 

\begin{theorem}[]\label{lem:schur_concavity_full}
Let $(\rho, \sigma)$ and $(\rho', \sigma')$ be  {two pairs of commuting quantum states, with $\sigma, \sigma'$ both full-rank.}
If  {$(\rho, \sigma) \succ (\rho', \sigma')$}, 
then it holds that $M_{s_{\min}}(\rho'\| \sigma') \geq M_{s_{\min}}(\rho\|\sigma)$, where
\begin{equation}
	M_x(\rho \| \sigma) := V(\rho\|\sigma) + \left(\frac{1}{\ln(2)} - \log(x) - S(\rho\|\sigma)\right)^2,
\end{equation} 
and $s_{\min}$ denotes the smallest eigenvalue of $\sigma$.
\end{theorem}

Result~\ref{res:monotone} follows by setting $\sigma = \sigma' =  \id$. 
Furthermore, since $\sigma$ is the minimum of the pre-order $\succ_\sigma$, monotonicity implies that
\begin{equation}
	0 \leq M_{s_{\min}}(\rho \| \sigma) \leq M_{s_{\min}}(\sigma\| \sigma ) = \left[\frac{1}{\ln(2)} -\log(s_{\min})\right]^2
\end{equation} 
for any state  $\rho$ that commutes with $\sigma$.

We now derive the following corollary of the above theorem as a general version of Result~\ref{res:bound_entropy_production}, where we write $\Delta S = S(\rho\| \sigma) - S(\rho'\| \sigma')$ and $\Delta V = V(\rho\| \sigma) - V(\rho'\|\sigma')$:
\begin{corollary} \label{cor:bound_relative_entropy_production}
     {Let $(\rho, \sigma)$ and $(\rho', \sigma')$ be pairs of commuting states, with $\sigma, \sigma'$ both full-rank.
    If $(\rho, \sigma) \succ (\rho', \sigma')$}, 

    then 
    \begin{equation} \label{eq:bound_relative_entropy_production}
        \Delta S\geq \frac{\Delta V}{2 \sqrt{ {M_{s_{\min}}}(\rho\| \sigma)}} \geq \frac{\Delta V}{2(1/\ln(2) - \log(s_{\min}))}. \end{equation}
\end{corollary}
\begin{proof}
	By monotonicity of the relative entropy and positivity of $M$, the statement is trivially true whenever $\Delta V \leq 0$. Hence, assume that $\Delta V > 0$. By Theorem~\ref{lem:schur_concavity_full}, we know that
	\begin{equation}
		 M_{s_{\min}}(\rho'\| \sigma') \geq M_{s_{\min}}(\rho\|\sigma).
	\end{equation}
	 Let us write $a = 1/\ln(2) - \log(s_{\min})$, so that
	 \begin{equation}
	 	 M_{s_{\min}}(\rho\|\sigma) = V(\rho\|\sigma) + (a-S(\rho\|\sigma))^2.
	 \end{equation} 
	Then inserting the definition of $M_{s_{\min}}$ and reshuffling terms yields (using $x^2-y^2 = (x-y)(x+y)$)
    \begin{align}
		0 & \leq \Delta S\cdot \left[2a - S(\rho \|\sigma) - S(\rho' \|\sigma')\right] - \Delta V \\
		& = \Delta S\cdot \left[2a - 2 S(\rho \|\sigma) + S(\rho \|\sigma) - S(\rho' \|\sigma')\right] - \Delta V\\
		& = \Delta S\cdot \left[2a - 2 S(\rho \|\sigma) + \Delta S \right] - \Delta V\\
		& = (\Delta S)^2 + 2\chi(\Delta S) - \Delta V,                                            
\end{align}
	where we write $\chi = a - S(\rho \| \sigma) \geq 0$. Solving the quadratic equation in $\Delta S$ then gives
    \begin{align}
        \Delta S & \geq -\chi + \sqrt{\chi^2 + \Delta V}                                                                    \\
                 & = - \sqrt{\chi^2}+ \sqrt{\chi^2 + \Delta V}                                                              \\
                 & \geq \frac{\Delta V}{2 \sqrt{\chi^2 + \Delta V}}                                                         \\
                 & \geq \frac{\Delta V}{2 \sqrt{\chi^2 + V(\rho \| \sigma)}} = \frac{\Delta V}{2 \sqrt{ {M_{s_{\min}}}(\rho\|\sigma)}}. 
    \end{align}
    Here, we have used the fact that $\Delta S \geq 0$ by monotonicity in the first step (to disregard one solution), positivity of $\chi$ in the second step and the concavity of the square root in the third step (more precisely that $f(y) \geq f(x) + f'(y)(y-x)$ for any differentiable concave function). This concludes the proof.
\end{proof}

We note in passing that such lower bounds on the production of relative entropy are essential for quantifying irreversibility in thermodynamics, where $\frac{1}{\beta}S(\rho \| \tau_\beta)$ denotes the non-equilibrium free energy of a system in state $\rho$ in an environment of inverse temperature $\beta$. Here, $\tau_\beta$ denotes the Gibbs state of the system at inverse temperature $\beta$. We leave the detailed investigation of applications of our results to thermodynamics for future work.

Next, we present the generalized version of Result~\ref{res:lower_bound_bipartite_unital}.  {Let $S$ and $E$ be two systems of respective dimension $d_S$ and $d_E$ and $\sigma \equiv \sigma_S \otimes \sigma_E$ as well as $\sigma' \equiv \sigma'_S \otimes \sigma'_E$ be two fixed product states on the joint system $SE$. Let $\mc E: \mc D(S \otimes E) \to \mc D(S \otimes E)$ be a quantum channel such that $\mc E(\sigma) = \sigma'$. Using this channel we can define a channel $\mc C: \mc D(S) \to \mc D(S)$ as 
\begin{equation} \label{eq:dilation_sigma_preserving}
    \mc C(\cdot) = \tr_E [\mc E(\cdot \otimes \rho_E)],
\end{equation}
for an initial state $\rho_E$ of the environment. As in the previous section, for some initial state $\rho_S$ on $S$, we denote as $\rho'_S = \mc C(\rho_S)$ the final state on $S$, as $\Delta S_S = S(\rho_S \| \sigma_S) - S(\rho'_S\|\sigma_S')$ the marginal change of relative entropy on $S$, and similarly for $\Delta V_S$ and the environment $E$. Finally, $I_{S:E}$ is the mutual information of the final state $\mc E(\rho_S \otimes \rho_E)$, as defined in Eq.~\eqref{eq:def_mutual_information} (with $\mc U$ replaced by $\mc E$). We then have the following:}
\begin{theorem}[] \label{thm:lower_bound_bipartite_relative}
     {Let $\mc C$ be a channel defined via Eq.~\eqref{eq:dilation_sigma_preserving}, for fixed and full-rank states $\sigma$, $\sigma'$ as well as $\rho_E$. Then, for any initial system state $\rho_S$ such that $\rho_S \otimes \rho_E$ commutes with $\sigma$ and $\mc E(\rho_S \otimes \rho_E)$ commutes with $\sigma'$, we have }

    \begin{equation}
	    \Delta S_S + \Delta S_E \geq \frac{\Delta V_S + \Delta V_E - K \cdot f(I_{S:E})}{2\sqrt{ {M_{s_{\min}}}(\rho_S \otimes \rho_E \| \sigma)}},
    \end{equation}
    where
    \begin{equation}
        K = \sqrt{2\ln(2)}\left(12 + \log^2(s_{\min}) + 8 \log^2(d_S \cdot d_E)\right)                        
    \end{equation}
    and $f(x) = \max\{\sqrt[4]{x}, x^2\}$. Here, $s_{\min}$ is the smallest eigenvalue of $\sigma$.
\end{theorem}      
\begin{proof}
    Applying Corollary~\ref{cor:bound_relative_entropy_production} with $\rho \equiv \rho_S \otimes \rho_E$ and Lemma~\ref{lem:subadditivity_variance} yields
    \begin{align}
	    S(\rho_S \otimes \rho_E\| \sigma) - &S(\mc C(\rho_S \otimes \rho_E \| \sigma)) \\ &\geq \frac{\Delta V_S + \Delta V_E - c \cdot f(I_{S:E})}{2\sqrt{M(\rho_S \otimes \rho_E \| \sigma)}}.  
    \end{align}
	The statement then follows from the fact that, for any state $\rho_{SE}'$ on $SE$ with mutual information $I_{ {S:E}}$, 
    \begin{align}
	    S(\rho_{SE}' \| \sigma') & = S(\rho_S'\| \sigma_S') + S(\rho_E'\| \sigma_E') + I_{ {S:E}} \\
				& \geq S(\rho_S'\| \sigma_S') + S(\rho_E'\| \sigma_E').             
    \end{align}
\end{proof}

\subsection{Relative entropy from local monotonicity} 
\label{sub:axiomatic_characterization_of_relative_entropy}
Lastly, let us discuss the general version of local monotonicity, which uniquely characterizes the relative entropy. 
To do this, let $\mc F$ be the set of all finite-dimensional density matrices with full rank and let $\mc C_{\mc F}$ be the set of quantum channels that map states of full rank to states of full rank (on possibly different Hilbert spaces), symbolically $\mc C_{\mc F}(\mc F)\subseteq \mc F$. We further generalize the notion of local monotonicity to states $\sigma_1\otimes \sigma_2$ that are not fixed points of a given channel: We say that a function $f$ on pairs of quantum states $(\rho,\sigma)$, with $\rho$ defined on the same Hilbert-space as $\sigma\in \mc F$, is \emph{locally monotonic with respect to $\mc C_{\mc F}$} if $C[\sigma_1\otimes \sigma_2]=\sigma_1'\otimes \sigma_2' \in \mc F$  for $C\in \mc C_{\mc F}$ and $\sigma_1\otimes \sigma_2 \in \mc F$ implies
\begin{align}
    f(\rho_1,\sigma_1) + f(\rho_2,\sigma_2) \geq f(\rho_1',\sigma_1') + f(\rho_2',\sigma_2'), 
\end{align}
where again $\rho_1'=\tr_2[C(\rho_1\otimes\rho_2)]$ and similarly for $\rho_2'$. We then have the following theorem. 
\begin{theorem}\label{thm:relative_local_monotonicity}
    Let $f$ be a function that is locally monotonic with respect to $\mc C_{\mc F}$ and assume that $\rho\mapsto f(\rho,\sigma)$ is continuous for fixed $\sigma\in \mc F$. Then
    \begin{align}
        f(\rho,\sigma) = a S(\rho \| \sigma) + b, 
    \end{align}
    where $a$ and $b$ are constants. 
\end{theorem}
The proof can be found in Appendix~\ref{app:single_out_entropy}.

\section{Conclusions and Outlook} 
\label{sec:conclusions}
In this work we comprehensively studied formal properties of the variance of (relative) surprisal together with their applications to single-shot (quantum) information theory. Before closing, let us comment on the high-level motivation for this work and open avenues for further research. To do this, we restrict again to the case of unital channels ($\sigma=\mathbbm{I}$) for simplicity.
                    
As discussed throughout the paper, the von~Neumann entropy quantifies information theoretic tasks in the asymptotic limit. Conversely, the min- and max-entropies typically appear in the fully single-shot regime. All these quantities are special cases of the R\'enyi entropies
\begin{align}
    S_\alpha(\rho) = \frac{1}{1-\alpha}\log(\tr[\rho^\alpha]), \qquad \alpha \in (0,1) \cup (1,\infty), 
\end{align}
with $S(\rho) = S_1(\rho) :=\lim_{\alpha\rightarrow 1}S_\alpha(\rho)$, $S_{\min}(\rho) := \lim_{\alpha\rightarrow \infty}S_\alpha(\rho)$ and $S_{\max}(\rho) := \lim_{\alpha\rightarrow 0}S_\alpha(\rho)$. Indeed, we can consider the min- and max-entropies to be the end  points of the ``R\'enyi curve'' $\alpha\mapsto S_\alpha(\rho)$. This curve encodes the full spectrum of a state (see below). Hence, roughly speaking, one can say that in the single-shot regime the full shape of the curve matters, while in the asymptotic i.i.d.~limit only the point $\alpha = 1$ matters. From this point of view the approach to single-shot information using min- and max-entropies rests on the observation that the end points of the (smoothed) curve capture many of the operationally relevant single-shot effects for a given state.
                    
In contrast, the approach presented here quantifies single-shot effects by studying not the end points but rather the neighbourhood of the R\'enyi curve around $\alpha = 1$. To see this,
consider the Taylor-expansion of $S_\alpha(\rho)$ around $\alpha=1+x$.
Then first performing the expansion and finally taking the limit $x\rightarrow 0$ yields
\begin{align}\label{eq:taylor}
    S_{\alpha}(\rho) = \sum_{n=1}^\infty \frac{\kappa^{(n)}}{n!}(1- \alpha)^{n-1}, 
\end{align}
where $\kappa^{(n)}$ is the $n$-th \emph{cumulant of surprisal}. In Appendix~\ref{app:relation_between_renyi_curve_and_cumulants_of_surprisal}, we present the definition of the cumulants of surprisal as well as the derivation of~\eqref{eq:taylor} (also see~\cite{Tomamichel2016a}). 
                    
Eq.~\eqref{eq:taylor} is interesting for a number of reasons. To begin with, we have $\kappa^{(1)} = S(\rho)$ and $\kappa^{(2)} = V(\rho)$. Hence,~\eqref{eq:taylor} shows that the variance of surprisal is (up to a factor of $-2$) the slope of the R\'enyi curve at $\alpha=1$ and gives the first order correction to the approximation $S_\alpha(\rho) \approx S(\rho)$. This fact is well-known~\cite{Tomamichel2016a,Song2001}. It lets us apply some of our results to the R\'enyi curve. For instance, Result~\ref{res:bound_min_max_by_variance} relates the neighbourhood of the R\'enyi curve around $\alpha = 1$ to its smoothed end-points, while Result~\ref{res:bound_entropy_production} constraints the possible changes of the Re\'nyi curve under unital channels, in the sense that the slope at $\alpha = 1$ can, by means of such channels, only be flattened at the expense of raising the curve at this point.

More generally, the expansion \eqref{eq:taylor} is interesting because it implies that the higher order cumulants of surprisal give a hierarchy of increasingly fine-grained knowledge about a state's spectrum. This follows once we recognize that for a $d$-dimensional state $\rho$, it suffices to know $S_n(\rho)$ for $n=2,\ldots,d$ to fully reconstruct the spectrum of $\rho$ (for the reader's convenience we provide a proof of this statement in Appendix~\ref{app:spectrumRenyis}).
In turn, the results of this paper, in which we studied the first order of this hierarchy, then suggest that studying the single-shot properties of higher order cumulants or surprisal could yield insights about single-shot information theory that are somewhat complementary to the approach of smoothed R\'enyi entropies. 
                    
In particular, it would be interesting whether it is possible to construct a hierarchy of Schur-concave functions with increasing relevance at the single shot level from cumulants of surprisal. 
As a first step in this direction, the following interesting problem arises: We have mentioned that knowing the R\'enyi entropies $S_n(\rho)$ for $n=2,\ldots,d$ provides full information about the spectrum of the state. 
Is it also true that the first $d-1$ cumulants of surprisal encode the full spectrum of the state?
Another problem to consider is the extension of our study to the fully quantum setting of non-commuting matrices. We leave these questions to future work.

{\bfseries Acknowledgements.} The authors would like to thank Angela Capel, Xavier Coiteux-Roy, Iman Marvian, Renato Renner, Carlo Sparaciari, Marco Tomamichel and Stefan Wolf for stimulating discussions and suggestions and especially Jens Eisert for fruitful comments on an earlier version of this work. We want to particularly thank Mark Wilde for suggesting the extension of our work to generic channels instead of $ \sigma $-preserving ones and Tony Metger and Raban Iten for the proof idea for showing that R\'enyi entropies are smooth functions of $\alpha$. P.~B.~and N.~N.~acknowledge support by DFG grant FOR 2724 and FQXi. P.~B. further acknowledges funding from the Templeton Foundation. N.~N. further acknowledges the Alexander von Humboldt foundation and the Nanyang Technological University, Singapore under its Nanyang Assistant Professorship Start Up Grant. H.~W. acknowledges contributions from the Swiss National Science Foundation via the NCCR QSIT as well as project No. 200020\_165843.

\bibliographystyle{apsrev4-1}
\bibliography{Compression}

\newpage
\onecolumngrid
\appendix
\section{Overview of Appendix}
\noindent Appendix~\ref{app:auxiliary_lemmata} establishes all the notation used throughout our proofs and collects a handful of technical lemmas. \\
	Appendices~\ref{app:proof_uniform_continuity} and ~\ref{app:proof_of_correction_to_subadditivity_of_relative_variance_} present the proofs for Lemma~\ref{lem:uniform contintuity relative variance} and Lemma~\ref{lem:subadditivity_variance} respectively. \\
	Appendix~\ref{app:lemma_variance_steepflat} presents the proof of the central technical Lemma~\ref{lem:variance_in_steepflat_properties} that underlies all of our sufficiency results.\\
	Appendix~\ref{app:second-order asymptotics} presents the application of Theorem~\ref{res:suff_transition_relative} to the case of finite i.i.d.~sequences.\\ Appendix~\ref{app:proof_schur_concavity} then provides the proof of Theorem~\ref{lem:schur_concavity_full}. \\
	Appendix~\ref{app:single_out_entropy} discusses details on the axiomatic characterization of locally monotonic functions, including the proofs of Result~\ref{res:entropy_single} and Theorem~\ref{thm:relative_local_monotonicity}. \\
	Finally, Appendix~\ref{app:relation_between_renyi_curve_and_cumulants_of_surprisal} provides the details to the expansion Eq.~\eqref{eq:taylor} and Appendix~\ref{app:spectrumRenyis} sketches the proof that a state's spectrum can be inferred from the values of $d-1$ R\'{e}nyi entropies, as claimed in the conclusion.

\section{Notation and auxiliary lemmata} 
\label{app:auxiliary_lemmata}

In the following we will make frequent use of the following definitions: 
\begin{itemize}
    \item $L(\rho\|\sigma) := \tr(\rho(\log(\rho) - \log(\sigma))^2)$,
    \item $\chi(x\|q) := x \log^2(\frac{x}{q})$, defined for $ q > 0 $, and over the regime $x\in [0,1]$ by continuous extension,
    \item $\eta(x) := - x \log(x)$, defined over the interval $[0,1]$ by continuous extension,
    \item $h_b(x) := \eta(x) + \eta(1-x)$, the binary entropy,
    \item $[d] := \{1, \dots, d\}$.
\end{itemize}   
We also remind the reader that we use logarithms with base $2$, $\log = \log_2$.
Lemmas \ref{lem:bound_binary_entropy} - \ref{lem:domination} are technical tools used in the derivation of our results. We list them here for completeness.

\begin{lemma}[] \label{lem:bound_binary_entropy}
    For any $x \in [0,1]$,
    $ h_b(x) \leq 2 \ln(2) \sqrt{x(1-x)} $.
\end{lemma}

\begin{lemma}[Klein's inequality] \label{lem:klein}
    Let $\rho, \sigma$ be density operators. Then 
    $   S(\rho \| \sigma ) \geq 0 $
    with equality \emph{iff} $\rho = \sigma$.
\end{lemma}

We will also use the following generalization of the Fannes-Audenaert inequality, which is implied by Lemma 7 in~\cite{Winter2016}:

\begin{lemma}[Continuity of relative entropy (Lemma 7, \cite{Winter2016})] \label{lem:winter}
    Consider any full rank state $\sigma$ with $s_{\min} > 0$ denoting its smallest eigenvalue. Then, for any two states $\rho, \rho'$ such that $D(\rho, \rho') \leq \epsilon$, we have
    \begin{equation}
        |S(\rho\|\sigma)- S(\rho'\|\sigma)| \leq -\log(s_{\min}) \epsilon + (1+\epsilon)h_b\left(\frac{\epsilon}{1+\epsilon}\right).
    \end{equation}
\end{lemma}

\begin{lemma}[Pinsker inequality] \label{lem:pinsker}
    For quantum states $\rho,\sigma$ acting on the same Hilbert space, $ S(\rho\|\sigma) \geq \frac{1}{2\ln(2)}\norm{\rho - \sigma}_1^2 $
\end{lemma}

\begin{lemma}\label{lem:existence_stochastic_matrix}
    Let $\mc E$ be a quantum channel such that $\mc E(\sigma) = \sigma'$, for two states $ \sigma = \sum_i^d s_i \proj{i} $ and $ \sigma' = \sum_i^{d'} s_i' \proj{i} $ (note that $ d \neq d' $ in general). Furthermore, given a state $\rho = \sum_i^d p_i \proj{i}$, suppose that $ \rho' = \mc E(\rho) = \sum_i^{d'} q_i \proj{i} $ commutes with $ \sigma' $. Then there exists a right stochastic $d \times d'$ matrix $E$, that is, a matrix with all non-negative entries  each of whose rows sums up to $1$, such that 
    \begin{align}
        pE & = q, \\
        sE & = s'  
    \end{align}
    where $p,q,s,s'$ are simply vectors containing the eigenvalues of $ \rho,\rho',\sigma $ and $ \sigma' $ respectively.
\end{lemma}

\begin{lemma}[Cantelli-Chebyshev inequality] \label{lem:cantelli}
	Given a random variable $X$ with finite mean $\mu$, variance $\sigma^2 < \infty$ and $\lambda > 0$, 
	\begin{equation}
		\mathrm{Pr}\left(X- \mu \geq \lambda\right) \leq \frac{\sigma^2}{\sigma^2 + \lambda^2}.
	\end{equation}
\end{lemma}

\begin{lemma}[Domination for definite integrals] \label{lem:domination}
	Let $f, g$ be continuous functions. If $f(x) \geq g(x)$ in the interval $[a,b]$, then
	\begin{align}
		\int_a^b f(x) dx \geq \int_a^b g(x) dx. 
	\end{align}
\end{lemma}

\noindent Next, 4 technical lemmas are proven explicitly and used in later parts of our work. 

\begin{lemma}[] \label{lem:triangle_commutation}
	Let $\rho, \rho'$ and $\sigma$ be three quantum states such that $\rho$ and $\rho'$ both commute with $\sigma$ but not necessarily with another. Then there exists a unitary chanel $\mc U$ such that i) $    [\rho, \mc U(\rho')] = 0 $, ii) $        D(\rho, \mc U(\rho')) \leq D(\rho, \mc  {\rho'})$ and iii) $\mc U (\rho')$ also commutes with $\sigma$. 
\end{lemma}
\begin{proof}
	We write $\sigma = \oplus_i s_i \mathbbm{1}_i$, where $\mathbbm{1}_i$ is the identity operator in the $i$-th eigenspace of $\sigma$. 
	Similarly, we can write $\rho=\oplus_i \rho_i$ and $\rho'=\oplus \rho'_i$. We then have
	\begin{align}
		D(\rho,\sigma) = \sum_i D(\rho_i,s_i\mathbbm{1}_i),\ D(\rho',\sigma) = \sum_i D(\rho'_i,s_i\mathbbm{1}_i),\ D(\rho,\rho') = \sum_i D(\rho_i,\rho'_i). 
	\end{align}
	The mapping $\rho' \mapsto U \rho' U^\dagger =: \mc U(\rho')$, with $U=\oplus_i U_i$ a block-diagonal unitary, is a $\sigma$-preserving quantum channel. Now, without loss of generality, choose a basis $\ket{i,j}$ in each eigenspace of $\sigma$ such that $\rho_i = \sum_j p_{i,j} \proj{i,j}$ with $p_{i,j} \geq p_{i,j+1}$. We can then choose $U_i$ so that $U_i \rho'_i U_i^\dagger = \sum_{i} p'_{i,j} \proj{i,j}$ with $p'_{i,j}\geq p'_{i,j+1}$ being the ordered eigenvalues of $\rho'_i$. Then clearly $[\rho_i,U_i\rho_i' U_i^\dagger]=0$. Furthermore, collecting the respective eigenvalues in vectors $\mathbf p_i, \mathbf p_i'$, from Theorem 4 in \cite{Markham2008} we directly find
	\begin{align}
		D(\rho,\mc U(\rho')) & = \sum_i D(\rho_i U_i\rho_i' U_i^\dagger) = \sum_i D(\mathbf p_i,\mathbf p_i') \\
		& \leq \sum_i D(\rho_i,\rho_i') = D(\rho,\rho').           \vspace{-0.6cm}                      
	\end{align}
\end{proof}

\begin{theorem}[Relative-majorization and Lorenz curves]\label{thm:Renes}
	Let $ (p,q) \in \mathbb{R}^n_+ $ and $ (p’,q’) \in \mathbb{R}^m_+ $ be two pairs of probability vectors (non-negative, normalised) such that when $ q_i = 0, $ then $ p_i = 0 $ and similarly for $ p’ $ and $ q’ $. Furthermore, denote the states $ \rho,\sigma $ to be diagonalized in the same basis with eigenvalues $ p,q $; while $ \rho',\sigma' $ are states diagonalized in the same basis with eigenvalues $ p',q' $ respectively. Then the following are equivalent: 
	\begin{enumerate}
		\item There exists a $ n \times m$-right stochastic matrix $ M $ such that $pM = p’ $ and $ qM=q’ $,
		\item $ \mathcal{L}_{\rho|\sigma} \geq  \mathcal{L}_{\rho'|\sigma'} $,
		\item For every continuous concave function $ g : \RR \to \RR$, 
		\begin{equation}
			\sum_{i=1}^n q_i' g\left(\frac{p_i'}{q_i'}\right)  \geq  \sum_{i=1}^m q_i g\left(\frac{p_i}{q_i}\right). 
		\end{equation} 
	\end{enumerate}
\end{theorem}

\begin{proof}
	Forms and proofs of this theorem appear in various places of the literature, and can be found in e.g. \cite{ruch1978mixing,RUCH1980222,blackwell1953,Marshall2011}. But since the proven statements often vary in some detail, we here present the proof for the reader's convenience.  
	
	$1. \Rightarrow 3.:$ Let $m_{ij}$ denote the elements of the matrix $M$ that exists by assumption. Define the matrix $A$ with elements $a_ij = q_i m_ij (q’_j)^{-1}$. It is easy to check that $A$ is a left-stochastic matrix and $Aq' = q$ and that $(p/q) A = (p'/q')$, where $p/q$ denotes element-wise division. Since $p/q$ and $p'/q'$ are real vectors due to fact that $q=0$ implies $p=0$, and similarly for the final pair, we can now apply Proposition A.1. p.579 from \cite{Marshall2011} to arrive at the desired statement.
	
	$3. \Rightarrow 2.:$ First, we note that, by definition of the Lorenz curve, $\mathcal{L}_{\rho|\sigma}(x)  = \mathcal{L}_{\rho'|\sigma'}(x)$ for $x = 0,1$. We hence need to show domination of the Lorenz curve only for $x \in (0,1)$. For this, let $(g_t)_{t \in \RR}$ and $(h_t)_{t \in \RR}$ be families of parametrized functions defined by $g_t(x) := \max\{x - t,0\}$ and $h_t(x) := \max\{t-x,0\}$. By convexity of the max function, $g_t$ as well as $h_t$ are convex for every $t$. Next, define the function $S_{\rho | \sigma}: (0,1) \to \RR$ as
	\begin{align}
		S_{\rho | \sigma}(x) &:= p_i/q_i, \quad \sum_{j = 0}^{i-1} q_j < x \leq \sum_{j = 0}^{i} q_j,
	\end{align}
	where $q_0 = 0$ and we enumerate the vectors such that $p_i/q_i \geq p_{i+1}/q_{i+1}$, without loss of generality. A moment's thought will show that $S_{\rho | \sigma}(x)$ is the value of the slope of the Lorenz curve $ \mathcal{L}_{\rho|\sigma}$ at $x$, whenever this slope is well defined. Now, for fixed $x$, we distinguish between the case $S_{\rho | \sigma}(x) \geq S_{\rho' | \sigma'}(x) $ and $S_{\rho | \sigma}(x) < S_{\rho' | \sigma'}(x) $. In the first case, we evaluate condition 3. for $g_t$ with $t = S_{\rho | \sigma}(x)$. This yields
	
	\begin{align}
		\mathcal{L}_{\rho | \sigma}(\sum_{i=0}^k q_i) - t \cdot \sum_{i=0}^k q_i \geq \mathcal{L}_{\rho | \sigma}(\sum_{i=0}^{k'} q'_i) - t \cdot \sum_{i=0}^{k'} q'_i,
	\end{align} 
	
	where $k$ is the largest index such that $p_k/q_k > t$, and similarly $k'$ is the largest index such that $p'_{k'}/q'_{k'} > t$. Geometrically, we can interpret the above inequality as making a statement about two parallel lines with slope $t$ running through the respective points  $\mathcal{L}_{\rho | \sigma}(\sum_{i=1}^k q_i) $ (on the LHS) and $\mathcal{L}_{\rho | \sigma}(\sum_{i=1}^{k'} q_i)$ (on the RHS). Now, by our choice of $t$, the point $(x,\mathcal{L}_{\rho|\sigma}(x))$ lies on the left hand parallel, while by the fact that  $S_{\rho | \sigma}(x) \geq S_{\rho' | \sigma'}(x) $, the point $(x, \mathcal{L}_{\rho'|\sigma'})$ is guaranteed to lie below the right hand-parallel. Hence, we have $\mathcal{L}_{\rho|\sigma}(x) \geq  \mathcal{L}_{\rho'|\sigma'}(x)$, as required.
	
	We now turn to the second case: If $S_{\rho | \sigma}(x) < S_{\rho' | \sigma'}(x) $, then we instead evaluate condition 3. for $h_t$ with $t = S_{\rho' | \sigma'}(x)$. This yields, after canceling out some terms, 
	\begin{equation}
		\mathcal{L}_{\rho | \sigma}(\sum_{i=0}^k q_i) - t \cdot \sum_{i=0}^k q_i \geq \mathcal{L}_{\rho | \sigma}(\sum_{i=0}^{k'} q'_i) - t \cdot \sum_{i=0}^{k'} q'_i,
	\end{equation} 
	where in turn this time $k$ is the largest index such that $p_k/q_k \geq t$, and similarly $k'$ is the largest index such that $p'_{k'}/q'_{k'} \geq t$. By following the same argument as in the first case, we know that the point $(x,\mathcal{L}_{\rho'|\sigma'}(x))$ lies on the right-hand parallel, while the point $(x, \mathcal{L}_{\rho'|\sigma'})$ is guaranteed to lie above the left-hand parallel. Hence, we have $\mathcal{L}_{\rho|\sigma}(x) \geq  \mathcal{L}_{\rho'|\sigma'}(x)$, as required.
	
	$2. \Rightarrow 1.: $ See Eq. 16 of \cite{renes2016relative} for the statement, where \cite{ruch1978mixing} contains the proof.
\end{proof}

    Note that under the condition $ [\rho,\sigma]=[\rho',\sigma']=0 $, the existence of a quantum channel  {$ \mathcal{E}$} such that $ \mathcal{E}(\rho) \in\mathcal{B}^\epsilon (\rho') $, is equivalent to having a classical stochastic matrix $ M $ such that $ pM = p' $ and $ qM = q' $ (by Lemma \ref{lem:existence_stochastic_matrix}). Theorem \ref{thm:majorization_vs_Lorenz_curves} is therefore implied by the first and second statements of Theorem \ref{thm:Renes}. 

\begin{lemma}[] \label{lem:bound_by_value_of_function_eta}
    Let $x,y \in [0,1]$ such that $|x-y| \leq \frac{1}{2}$. Then 
    \begin{align}
        |\eta(x)- \eta(y)| \leq \eta(|x-y|). 
    \end{align}
\end{lemma}

\begin{proof}
    The statement is clearly true if $x=y$. Hence, w.l.o.g. let $x > y$, and set $z = |x-y| =x-y \leq 1/2$. First, note that  
    \begin{align}
        |\eta(x)- \eta(y)| = \left| \int_0^{z} \eta'(y + r) dr \right| =: F_z(y), 
    \end{align}
    and that $ F_z(0) = \eta(z) $, so that is sufficient to show that,
    \begin{align}\label{eq:F_z_max}
        F_z(0) \geq F_z(y). 
    \end{align}
    To show this, let us begin by evaluating the derivative 
    $   \eta'(x) = \frac{-1}{\ln(2)}[\ln(x) + 1]$, 
    and noting that this is a monotonically decreasing function, with a root at $ x^* = e^{-1} $. As graphically shown in Fig.~\ref{fig:etaprime}, Eq.~\eqref{eq:F_z_max} states that of all integrals of fixed width $ z $, the one with the largest absolute value is the one over the interval $ [0,z]$.
    \begin{figure}[h!]
        \includegraphics[width=0.5\linewidth]{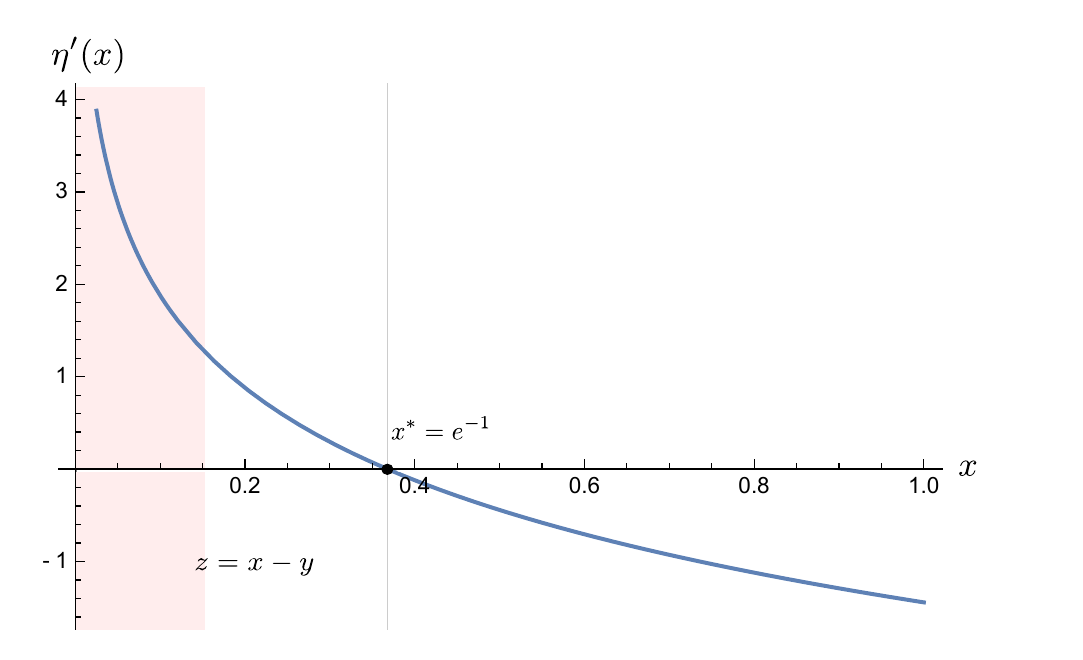}\caption{The function $ \eta'(x) $ and the corresponding area for the integral $ F_z(0) $. }\label{fig:etaprime}
    \end{figure}
                                
    We now show that this is indeed the case. We consider three different cases: If $ y\leq x^* - z $, the statement is automatically true because $\eta'$ is positive and monotonically decreasing below $x^*$, so that an application of Lemma~\ref{lem:domination} yields 
    \begin{equation}
        F_z(0) = \int_0^{z} \eta'( r) dr \geq \int_0^z \eta'(y+r) dr = F_z(y).
    \end{equation}
                                
    From similar reasoning, we know that $F_z(y) \leq F_z(1-z)$ whenever $x^* \leq y \leq 1-z$, by monotonicity and negativity of $\eta'$ above $x^*$. For the third case, when $x^* - z < y < x^*$, the fact that part of the integral is positive and part negative implies that 
    \begin{equation}
        F_z(y) \leq \max\{F_z(x^* - z), F_z(x^*)\} \leq \max\{F_z(0), F_z(1-z)\},
    \end{equation}
    where in the second step we applied the bounds derived for the previous cases. 
    Hence, it remains to show that $ F_z(0) \geq F_z(1-z) $. 
    To do so, first note that $ F_z(0) = \eta(z)$ and $ F_z(1-z)= \eta(1-z) $.
    Furthermore, the function $g(z) := \eta(z) - \eta(1-z)$ is continuous over $[0,1]$, is positive at $z=e^{-1} \in [0, 1/2]$, with roots at $x=0,1/2,1$. By invoking the intermediate value theorem, we know that $ g(z)\geq 0 $ for all $ z \leq 1/2 $, which concludes the proof.
\end{proof}

\begin{lemma}[] \label{lem:bound_by_value_of_function_chi}
    Let $q \in (0,1]$ and $x,y,\in [0,q]$. If $|x-y| \leq q/e^2$, then
        \begin{align}
            |\chi(x\|q)- \chi(y\|q)| \leq \chi(|x-y|\|q). 
        \end{align}
    \end{lemma}
                            
    \begin{proof}
        The proof is in spirit very similar to that of Lemma~\ref{lem:bound_by_value_of_function_eta}. Again the statement is trivially true if $x = y$. Assume, then, again without loss of generality that $x > y$ and set $z = |x-y| = x-y \leq q/e^2$. We note that  
        \begin{align}
            |\chi(x\|q)- \chi(y\|q)| = \left| \int_0^{z} \chi'(y + r\|q) dr \right| =: G_z(y\|q), 
        \end{align}
        and $G(0\|q) = \chi(z\|q)$, so it is sufficient to show
        $   G_z(0\|q)\geq G_z(y\|q) $. 
        Now, with the same strategy, let us first evaluate
        \begin{align} \label{eq:chi_derivative}
            \chi'(x\|q) = \frac{1}{\ln(2)^2}\left[2\ln(x/q) + \ln^2(x/q)\right], 
        \end{align}
        and plot it in Fig.~\ref{fig:chiprime}.
        \begin{figure}[h!]
            \includegraphics[width=0.5\linewidth]{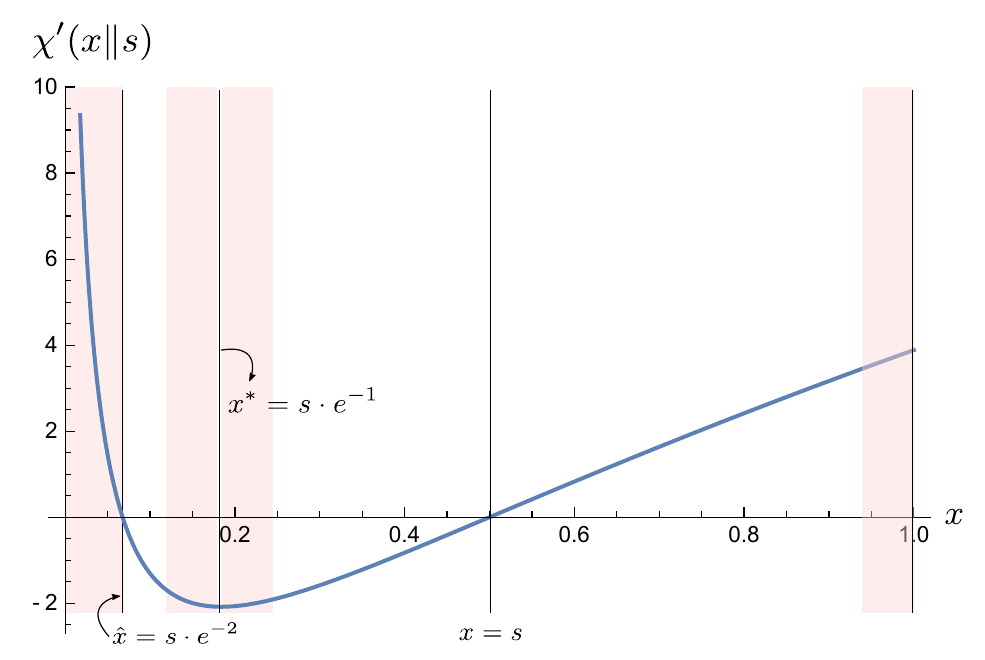}\caption{The function $ \chi'(x\|q) $ for $q = 0.5$, with the relevant integrals highlighted in red.}\label{fig:chiprime}
        \end{figure}                                                                                                                                                    
        We are interested in three intervals of this function:
        : $[0,q/e^2]$, where it is monotonically decreasing and positive on the interval; $[q/e^2,q/e]$, where it is monotonically decreasing and negative; and $[q/e,q]$, where it is monotonically increasing and negative.  
        By monotonicity on these separate intervals, the fact that $z \leq q/e^2$ (which implies that each of these intervals is at least as wide as $z$) and invoking Lemma~\ref{lem:domination}, it follows that  
        \begin{align}
            \max_{r \in [0,s-z]} G_z(r\|q) = \max\lbrace G_z(0\|q), G_z(q/ e-z\|q) + G_z(q/e\|q) \rbrace, 
        \end{align}
        by the following reasoning: For $r \leq q/e^2 - z$, by applying Lemma~\ref{lem:domination}, we have that $G_z(r \|q) \leq G_z(0\|q)$ by positivity in that interval. Similarly, we can bound $G_z(r\|q)$ for the values $q/e^2 \leq r \leq s-z$ by the second term in the above bracket. Finally, for $q/e^2 - z < r < q/e^2$, parts of the integral cancel out, so that we can bound the integral by the above two terms. It hence remains to show that $G_z(0 \|q)$ always dominates the second term above, which we check by explicit evaluation:  we first have
        \begin{equation}
            G_z(0\|q) = \chi (z\|q) = \frac{z}{\ln^2(2)}\cdot \ln^2(z/q),
        \end{equation}
        where, since $ z\leq q/e^2 \leq 1$ by assumption and the fact that $\ln^2(x)$ is strictly monotonically decreasing in the range $[0,1]$, $ \ln^2(z/q) \geq \ln^2(e^{-2}) = 4$.
        On the other hand, we can upper bound the second term by noting that 
        \begin{align} \label{eq:bound_minimum_middle_term}
            G_z(q/e - z\|q) + G_z(q/e\|q) & \leq 2z\cdot |\chi'(q/e\|q)|= \frac{2z}{\ln^2(2)}, 
        \end{align}
        which is always smaller than $G_z(0\|q)$. 
    \end{proof}
                            
    \section{Proof of uniform continuity of the relative variance (Lemma~\ref{lem:uniform contintuity relative variance})} %
    \label{app:proof_uniform_continuity}
                            
    The main result of uniform continuity of relative variance (and therefore the non-relative variance of surprisal) is proven in Lemma~\ref{lem:uniform contintuity relative variance}. To do so, let us first establish the following technical lemma. 
                            
    \begin{lemma}[] \label{lem:bound_squared_log_by_distance}
         {Let $\sigma$ be a positive-definite operator with $0 < \sigma \leq \id $ and smallest eigenvalue $s_{\min}$
on a $d$-dimensional Hilbert space with $d \geq 2$, and $\rho, \rho'$ be two states that both commute with $\sigma$.}
        If $D \equiv D(\rho, \rho') \leq 1/(2e^2)$, then
        \begin{equation}
            |L(\rho\|\sigma) - L(\rho'\|\sigma)| \leq c_1 D\log^2(d) + \chi(2D\|1) + 2 \eta(2D)\log(d),
        \end{equation}
        where $c_1 = 12 + \log^2(s_{\min}) + 8 \log^2(d)$. More generally, for any $D$,
        \begin{align}
            |L(\rho\|\sigma) - L(\rho'\|\sigma)| \leq c_1 D\log^2(d) + c_2\sqrt{D}, 
        \end{align}
        with $c_2 = 6 + 2\log(d)$.
    \end{lemma}
                            
    \begin{proof}
	    As the first step, we note that due to fact that both $\rho$ and $\rho'$ commute with $\sigma$, we only need to consider the spectra of the various states. This follows from Lemma~\ref{lem:triangle_commutation}. Let $\mc U$ denote the unitary channel defined in the proof of that Lemma. Then by construction we have that $[\rho, \mc U(\rho')] = 0$. Since
        \begin{equation}\label{eq:bound_trace_distance}
            D(\rho, \mc U(\rho')) \leq D(\rho, \rho')
        \end{equation}
        and 
        $ L(\mc U(\rho') \| \sigma) = L(\rho' \| \sigma) $,
        it follows that we can in the following replace  $\rho'$  by $\mc U(\rho')$, without loss of generality.
                                                            
        Since all three states then commute with another, we can assume the decompositions $\rho = \sum_i p_i \proj{i}, \rho' = \sum_i q_i \proj{i}$ and $\sigma = \sum_i s_i \proj{i}$, in terms of which we have 
        \begin{align} 
            D(\rho, \rho')                     & = \frac{1}{2}\sum_i |p_i - q_i| = \frac{1}{2}\sum_i x_i,\label{eq:bound_trace_distance_by_eigenvals}                     \\
            |L(\rho\|\sigma)-L(\rho'\|\sigma)| & =\Big| \sum_i \chi(p_i\|s_i)-\chi(q_i\|s_i) \Big| \leq \sum_i | \chi(p_i\|s_i)-\chi(q_i\|s_i)|,\label{eq:LHS_simplified} 
        \end{align}
        and where we have introduced the variable $ x_i := |p_i-q_i| $.  
        We now show that each of the terms in the RHS of \eqref{eq:LHS_simplified} can be upper bounded by a term of the form either $\chi(x_i\| s_i)$ or $C \cdot x_i$ for some constant $C$. To see this, consider the $i$th term in the sum and let us distinguish the following cases, where we assume without loss of generality that $q_i < p_i$:
                                                            
        \textit{Case I: $p_i \leq s_i/e^2$} In this case, we know that $x_i \leq s_i/e^2$ and so can apply Lemma~\ref{lem:bound_by_value_of_function_chi} to find that 
        \begin{equation}
            | \chi(p_i\|s_i)-\chi(q_i\|s_i)| \leq \chi(x_i\| s_i).
        \end{equation}
                                                            
        \textit{Case II: $s_i/e^2 \leq q_i$}. In this case, we can make use of the fact that $\chi(\cdot\|s_i)$ is Lipschitz continuous in its first argument over the interval $[s_i/e^2, 1]$. In particular, by differentiability of $\chi$ over this interval, we have that 
        \begin{align}
            |\chi(p_i\|s_i)- \chi(q_i\|s_i)| & \leq \sup_{r \in [s_i/e^2,1]} |\chi'(r\|s_i)|  \cdot |p_i-q_i|                                    \\
                                             & = \max \{|\chi'(s_i/e\|s_i)|, |\chi'(1\|s_i)|\}   \cdot x_i                             \\
                                             & = \frac{1}{\ln(2)^2}\cdot {\rm max}~ \{1,\ln^2(s_i) - 2\ln(s_i)\} \cdot x_i =: \tilde{C}(s_i)x_i, 
        \end{align}
        where $\tilde{C}(s_i)$ is the Lipschitz constant.
        
        \textit{Case III: $q_i < s_i/e^2 < p_i$} Here, we distinguish three sub-cases. To discuss these cases, we note that since for fixed $s_i$, $\chi(\cdot\| s_i)$ is continuous and has roots at $0$ and $s_i$, as well as a local maximum at $s_i/e^2$, by the mean value theorem there must be a point $q_i^* \in [s_i/e^2, s_i]$ such that 
        \begin{equation}
            \chi(q_i\| s_i) = \chi(q_i^*\|s_i).
        \end{equation} 
        We now distinguish the following sub-cases: First, assume that $q_i^* \leq p_i$. Then we can make use of Lipschitz continuity, since
        \begin{align}
            |\chi(p_i\|s_i)- \chi(q_i\|s_i)| & = |\chi(p_i\|s_i)- \chi(q_i^*\|s_i)|                                                                                  \\
                                             & \leq \tilde{C}(s_i)|p_i - q_i^*| \leq \tilde{C}(s_i)(p_i - q_i^*) \leq \tilde{C}(s_i)(p_i - q_i) = \tilde{C}(s_i)x_i. 
        \end{align}
        Note that this sub-case always covers situations in which $p_i \geq s_i$, because we are guaranteed that $q_i^* \leq s_i$. 
        Hence, it remains to consider the case $q_i^* < p_i < s_i$. Now, if $x_i \leq s_i/e^2$, then we can apply Lemma~\ref{lem:bound_by_value_of_function_chi} to find that 
        \begin{equation}
            | \chi(p_i\|s_i)-\chi(q_i\|s_i)| \leq \chi(x_i\| s_i).
        \end{equation}
        Finally, if $x_i > s_i/e^2$, then we have that 
        \begin{equation}
            | \chi(p_i\|s_i)-\chi(q_i\|s_i)|  = \chi(p_i\|s_i)-\chi(q_i\|s_i) < \chi(s_i/e^2\| s_i) = 4s_i/e^2 < 4x_i, 
        \end{equation}
        where we twice used the fact that $\chi$ is strictly monotonically decreasing on the interval $[s_i/e^2, s_i]$ and positive. In the first step, combined with the definition of $q_i^*$ this implies that $\chi(p_i\| s_i) > \chi(q_i \| s_i)$. In the second step, this implies that $\chi(p_i \|s_i) < \chi(s_i/e^2\| s_i)$.
                                                            
        Overall, we have seen that we can upper bound each of the terms on the RHS of \eqref{eq:LHS_simplified} by either $\chi(x_i \| s_i)$ or by $C(s_i) \cdot x_i$, where $\hat{C}(s_i) := \max\{4, \tilde{C}(s_i)\} = \max\{4, \chi'(1\|s_i)\}$. 
                                                            
        In fact, we can derive a further, simple upper bound for $\hat{C}(s_i)$ as follows: First we have that $\chi'(1\|q) \geq 4$ for values $q \leq \exp(1-\sqrt{1+4\ln^2(2)}) \approx 0.49$. At the same time, for values $q \leq \exp(-2/3) \approx 0.51$, we have that $\ln^2(q) - 2\ln(q) \leq 4\ln^2(q)$. This implies that we have the upper bound 
        \begin{equation}
            \hat{C}(s_i) \leq 4\cdot \max\{1, \log^2(s_i)\} =:C(s_i).
        \end{equation}
        Let $A$ denote the set of indices $i$ that we have bounded by $\chi(x_i \| s_i)$ and $B = [d]\backslash A$ those that we have bounded by $C(s_i) \cdot x_i$. We now turn to upper bound the two groups of terms corresponding to these two sets. In particular, let
        \begin{align}
            T_1 & :=\sum_{i \in A} |\chi(p_i\|s_i) - \chi(q_i\|s_i)|,\qquad 
            T_2:=\sum_{i \in B} |\chi(p_i\|s_i) - \chi(q_i\|s_i)| 
        \end{align}
        respectively, and denote $\Delta_1 := \sum_{i \in A} x_i, \Delta_2 := \sum_{i \in B} x_i$. We can straightforwardly bound $T_2$ as
        \begin{equation}
            T_2 \leq \sum_{i \in B}C(s_i) x_i \leq C(s_{\min}) \Delta_2 \leq 2 C(s_{\min}) D,
        \end{equation}
        where we recall that $D\equiv D(\rho, \rho')$. Bounding $T_1$ is more involved.
        By applying the previously derived upper bound, we have
        \begin{align} 
            T_1\leq \sum_{i \in A} \chi(x_i\|s_i). 
        \end{align}
        We first note
        \begin{align}
            \chi(x_i \| s_i) & = x_i\log^2(x_i) + 2\eta(x_i) \log(s_i) + x_i\log^2(s_i)          \\
                             & \leq x_i\log^2(x_i) + x_i\log^2(s_i)                              \\
                             & \leq x_i\log^2(x_i) + x_i\log^2(x_i e^2)                          \\
                             & = 2x_i\log^2(x_i) - 4\eta(x_i) + 4x_i \leq 2x_i\log^2(x_i) + 4x_i 
            \label{eq:bound_chi2} 
        \end{align}
        where in the second step we used that $\eta(x_i)$ is positive and $s_i \leq 1$, in the second $x_i \leq s_i/e^2$, which holds for all the terms in $T_1$ by the previous arguments, and in the last step again positivity of $\eta(x_i)$.
        Next, we make use of the identity 
        \begin{align}
            \log^2(x_i) = \log^2(x_i/\Delta_1) +2\log(x_i)\log(\Delta_1) - \log^2(\Delta_1). 
        \end{align}
        Plugging this into the RHS of \eqref{eq:bound_chi2} yields
        \begin{equation} \label{eq:expansion_upper_bound_T1}
            \sum_{i \in A}  \chi(x_i\|s_i) \leq 2\Delta_1\cdot F_1-2 F_2 + 4 \Delta_1 - \Delta_1 \log^2(\Delta_1),
        \end{equation}
        where 
        \begin{align}
            F_1 & := \sum_{i \in A}\frac{x_i}{\Delta_1}\log^2\left(\frac{x_i}{\Delta_1}\right) =  \sum_{i \in A} \chi\left(x_i/\Delta_{\leq}\|1\right), \qquad 
            F_2  := \log(\Delta_1) \cdot \sum_{i \in A} \eta(x_i).                                                                              
        \end{align}
        To bound $F_1$, we note that $ \lbrace x_i/\Delta_{\leq} \rbrace_i $ form a $|A|$-dimensional probability vector, corresponding to some density matrix $ \varrho $. Hence,
        \begin{equation}\label{f1}
            F_1 = L(\varrho\|\id) = V(\varrho\|\id) + S(\varrho\|\id)^2 \leq 4\log^2(|A|) \leq 4 \log^2(d),
        \end{equation}
        where we have used the fact that, by Property 6 of the variance, as presented in the main text,
        \begin{equation}
            V(\rho' \| \id) = V(\rho') \leq \frac{1}{4}\log^2(d-1) + 1/\ln^2(2) \leq 3 \log^2(d)
        \end{equation}
        and that $ S(\varrho\|\id)^2 = S(\varrho)^2 \leq \log^2(d)$ for a $d$-dimensional density matrix for the case $|A| \geq 2$. Clearly, this upper bound is also valid for $|A| \in \{0,1\}$. Next, to lower bound the term $ F_2 $, note that
        \begin{align}
            \eta(x_i) = \Delta_{\leq}\cdot \eta(x_i/\Delta_{\leq}) - x_i \log(\Delta_{\leq}), 
        \end{align}
        which yields
        \begin{align}
            F_2 & = \sum_{i \in A} \log(\Delta_1) \cdot [\Delta_1 \cdot \eta(x_i/\Delta_1)-x_i\log\Delta_1] \geq \min\{0, \Delta_1 \log(\Delta_1) \log(d) \}  - \Delta_1\log^2(\Delta_1). 
        \end{align}
        The minimization arises because we distinguish two cases: If $\Delta_1 \geq 1$, then the terms $\log(\Delta_1) \Delta_1 \cdot \eta(x_i/\Delta_1)$ are positive and can be lower bounded by zero, while if $\Delta < 1$, the term is negative and can be lower bounded by the $\Delta_1 \log(\Delta_1) \log(d)$, again using the fact that 
        the $\{x_i/\Delta_1\}_{i \in A}$ form a probability distribution. Plugging these bounds back into \eqref{eq:expansion_upper_bound_T1} then yields
        \begin{align}
            T_1 & \leq 4\Delta_1\cdot (2\log^2(d) + 4) + \chi(\Delta_1\|1) + \max\{0,2 \eta(\Delta_1)\log(d)\} \\
                & \leq 8D \cdot (2\log^2(d) + 4) + \chi(\Delta_1\|1) + \max\{0,2 \eta(\Delta_1)\log(d)\},      
        \end{align}
                                                            
        We are finally in a position to combine the bounds on $T_1$ and $T_2$. This gives 
        \begin{align}
            |L(\rho\| \sigma) - L(\rho' \| \sigma)| & \leq 2[C(s_{\min}) + 8\log^2(d) + 8] \cdot D + \chi(\Delta_1\|1) + \max\{0,2 \eta(\Delta_1)\log(d)\}                \\
                                                    & \leq 8(3+ \frac{1}{4}\log(s_{\min})^2 + 2\log^2(d))\cdot D + \chi(\Delta_1\|1) + \max\{0,2 \eta(\Delta_1)\log(d)\}, 
        \end{align}
        where we used 
        \begin{equation}
            C(s_{\min}) + 8\log^2(d) + 8 = \max\{4, \log^2(s_{\min})\} + 8\log^2(d) + 8 \leq 12 + \log^2(s_{\min}) + 8 \log^2(d) =: c_1.
        \end{equation}
        Now, if $D \leq 1/(2e^2)$, then we can use the monotonicity of $\chi(\cdot\|1)$ and $\eta$ over the interval $[0,1/e^2]$ to bound $\chi(\Delta_1 \| 1) \leq \chi(2D\| 1)$ and $\eta(\Delta_1) \leq \eta(2D)$. This provides the first statement of the lemma. 
                                                            
        For the second statement, in which we have no promise on the value of the trace distance between $\rho$ and $\rho'$, it suffices to note that, for $x \in [0,2]$,
        \begin{align}
            \eta(x) & \leq \sqrt{x},  \quad 
            \chi(x\|1) \leq 6\sqrt{x}. 
        \end{align}
        Since, $\Delta_1 \in [0,2]$, this implies
        \begin{equation}
            \chi(\Delta_1\|1) + \max\{0,2 \eta(\Delta_1)\log(d)\} \leq \left(e + 2\log(d)\right)\sqrt{\Delta_1} \leq c_2\sqrt{D},
        \end{equation}
        where we defined $c_2 := (6 + 2\log(d))$.           
    \end{proof}
                            
    With this technical lemma established, it is then relatively easy to prove uniform continuity of $V(\rho\|\sigma)$, which is stated as Lemma~\ref{lem:uniform contintuity relative variance} in the main text (we restate it here for convenience).
    
\setcounter{theorem}{9}
    \begin{relemma}[Uniform continuity of the relative variance] 
  		{Let $\sigma$ be a positive-definite operator with $0 < \sigma \leq \id $ and smallest eigenvalue $s_{\min}$
  			on a $d$-dimensional Hilbert space with $d \geq 2$, and $\rho, \rho'$ be two states that both commute with $\sigma$}. Then, 
  		\begin{align}
  			|V(\rho\|\sigma) - V(\rho'\|\sigma)| \leq 2K\sqrt{ {D(\rho, \rho')}}, 
  		\end{align}
  		where $K = 8 \log^2(d) + \log(d) + 2 \log^2(s_{\min}) - 4 \ln(2) \log(s_{\min}) + 15$.
  \end{relemma} 
    
    \begin{proof}
        We have
        \begin{align}
            |V(\rho\|\sigma) - V(\rho'\|\sigma)| & \leq |L(\rho\|\sigma) - L(\rho'\|\sigma)| + |S(\rho'\|\sigma)^2 - S(\rho\|\sigma)^2|                                                             \\
                                                 & \leq |L(\rho\|\sigma) - L(\rho'\|\sigma)| + (S(\rho'\|\sigma) + S(\rho\|\sigma))|S(\rho'\|\sigma) - S(\rho\|\sigma)|                             \\
                                                 & \leq  |L(\rho\|\sigma) - L(\rho'\|\sigma)| - 2 \log(s_{\min})|S(\rho'\|\sigma) - S(\rho\|\sigma)|                                                \\
                                                 & \leq  c_1 \epsilon + c_2\sqrt{\epsilon} - 2 \log(s_{\min})|S(\rho'\|\sigma) - S(\rho\|\sigma)|                                                   \\
                                                 & \leq  c_1 \epsilon + c_2\sqrt{\epsilon} - 2\log(s_{\min})\left((1+\epsilon)h_b\left(\epsilon/(1+\epsilon)\right) - \log(s_{\min})\epsilon\right) \\
                                                 & \leq  c_1 \epsilon + c_2\sqrt{\epsilon} - 2\log(s_{\min})\left(4\ln(2)\sqrt{\epsilon} - \log(s_{\min})\epsilon\right)                            \\   
                                                 & \leq  \left(c_1 + c_2 - 8\ln(2)\log(s_{\min}) + 2 \log^2(s_{\min})\right)\sqrt{\epsilon} =: K \sqrt{\epsilon},                                   
        \end{align}
        where we used $\max\{S(\rho\| \sigma), S(\rho'\| \sigma)\} \leq -\log(s_{\min})$ in the third step, Lemma~\ref{lem:bound_squared_log_by_distance} in the fourth step, Lemma~\ref{lem:winter} in the fifth step, Lemma~\ref{lem:bound_binary_entropy} in the sixth step and $\epsilon \leq \sqrt{\epsilon}$ in the last step (since $\epsilon \in [0,1]$.                                                                                                        
    \end{proof}

    \section{Proof of correction to subadditivity of relative variance (Lemma~\ref{lem:subadditivity_variance})} %
    \label{app:proof_of_correction_to_subadditivity_of_relative_variance_}

    For notational convenience, for the remainder of this appendix we write $V \equiv V(\rho\|\sigma)$ and $V_1 \equiv V(\rho_1\|\sigma_1)$ and similarly for the other subsystem and other quantities, $L$ and $S$. 
   \setcounter{theorem}{10}
\begin{lemma}[Correction to sub-additivity of relative variance.] 
		{Let $\rho, \sigma$ be two commuting quantum states on a $d$-dimensional, bipartite system, with $d \geq 2$ and $\sigma$ full-rank. If } $\sigma = \sigma_1\otimes\sigma_2 $ is a product state with smallest eigenvalue $s_{\min}$, then
		\begin{align}
			V(\rho\|\sigma) \leq V(\rho_1\|\sigma_1) + V(\rho_2\|\sigma_2) +  K'\cdot f(I_{1:2}), 
		\end{align}
		where $K' = \sqrt{2\ln(2)}\left(12 + \log(s_{\min})^2 + 8 \log^2(d)\right)$, $f(x) = \max\{\sqrt[4]{x}, x\}$ and $ I_{1:2} $ denotes the mutual information between the two partitions of $ \rho $.
\end{lemma}                            
    \begin{proof}
    	Let us begin by introducing the following shorthand notation:
    	\begin{itemize}
    		\item $I $ as the mutual information of $ \rho $ across the partitions 1 and 2, 
    		\item $D \equiv D(\rho, \rho_1 \otimes \rho_2)$,
    		\item $S_\otimes = S(\rho_1\otimes\rho_2\|\sigma) = S_1+S_2$,
    		\item $ L_\otimes = L(\rho_1 \otimes \rho_2\|\sigma) = L_1 + L_2 + 2S_1S_2$
    	\end{itemize} 
    Next, note that
    	\begin{equation}
    		S = I+ S_\otimes.
    	\end{equation}
    	The special case of this equality is well-known for the non-relative version of von Neumann entropy, i.e. $ I(A:B) = S(\rho_A)+S(\rho_B)-S(\rho_{AB})$, while the identity above for relative entropies is proved in Proposition 2 of \cite{capel2017superadditivity}.
        Therefore, we have
        \begin{align}
        		V = L - S^2 
        		& =L - L_1-L_2+L_1+L_2 - (I+S_\otimes)^2                 \\
        		& = L - L_1-L_2+L_1+L_2 - I^2 -2IS_\otimes - S_1^2  - 2S_1S_2- S_2^2\\
        		& = V_1+V_2+L-L_\otimes - I^2 -2IS_\otimes\\
        		& \leq V_1 + V_2 + \zeta_\rho,                           
        \end{align}
        where 
        $ \zeta_\rho:= L - L_\otimes $, since in the second last line, the terms $ I^2 $ and $ 2 I S_\otimes  $ are positive and can be dropped with the inequality. Our goal is to bound this remaining function in terms of the mutual information $ I $. To do so, we can apply Lemma~\ref{lem:bound_squared_log_by_distance}, since $[\rho, \sigma]= 0$ implies that $[\rho_1 \otimes \rho_2, \sigma] = 0$. Applying this lemma yields
        \begin{align}
		 \zeta_\rho & \leq |L-L_\otimes| \leq  c_1 D + c_2 \sqrt{D}  \\
		& \leq c'_1 \sqrt{I} + c'_2 \sqrt[4]{I},
\end{align}
        with $ c'_1 = \sqrt{2\ln 2} \cdot c_1 $ and $ c'_2 = \sqrt[4]{2\ln(2)}c_2 $, where $c_1, c_2$ are the constants from the statement of Lemma~\ref{lem:bound_squared_log_by_distance}, and where in the second step we used Pinsker's inequality (Lemma~\ref{lem:pinsker}) and the fact that $ I = S(\rho\|\rho_1\otimes\rho_2) $.
        Finally, by noting that, for $d\geq 2$, $c'_1 = {\rm max} \lbrace c'_1,c'_2\rbrace$, and $f(I) = \max\{\sqrt{I}, \sqrt[4]{I}\}$, we obtain the bound
        \begin{align} 
            \zeta_\rho & \leq  c'_1 f(I). 
        \end{align}
        The statement of the Lemma then follows by setting $K' = c'_1$.
    \end{proof}

    \section{Proof of Lemma~\ref{lem:variance_in_steepflat_properties}} 
    \label{app:lemma_variance_steepflat}
                            
    In order to proof the central technical result of Lemma~\ref{lem:variance_in_steepflat_properties}, we first make the following simple observation of lower and upper bounds for Lorenz curves, which is spelled out in Lemma~\ref{thm:intermediate_flat_state}.
    \setcounter{theorem}{27}
    \begin{lemma}\label{thm:intermediate_flat_state}
        Given any states $ \rho ,\sigma $ such that $ [\rho ,\sigma]=0 $ and $ {\rm supp}(\rho) \subseteq {\rm supp}(\sigma) $, express them in their common eigenbasis $ \rho = \sum_i p_i \ketbra{i}{i} $ and $ \sigma = \sum_i s_i \ketbra{i}{i} $, where $ \lbrace p_i\rbrace $ and $ \lbrace s_i\rbrace $ are the eigenvalues of $ \rho $ and $ \sigma $, respectively. Let $ l_c(x) = {\rm min} (c\cdot x, 1)$ for $ x\in 
        [0,1] $. Denote 
        \begin{equation}\label{someotherequation}
            r_{\rm min} = \min_{i\in P_+} \frac{p_i}{s_i} , \qquad r_{\rm max}  = \max_{i\in P_+} \frac{p_i}{s_i},
        \end{equation}
        where $ P_+ = \lbrace i|p_i >0\rbrace $ is the set of indices for which $ p_i $ is strictly positive. Then the Lorenz curve of $ \rho $ with respect to $ \sigma $ satisfies on the whole interval $ x\in[0,1] $:
        \begin{equation}\label{ell}
            \ell_{r_{\rm min}}(x)\leq \mathcal{L}_{\rho|\sigma}(x) \leq \ell_{r_{\rm max} }(x).
        \end{equation}
    \end{lemma}
    \begin{proof}
        This is obvious given the concavity of the Lorenz curve itself. 
    \end{proof}
                            
    \begin{remark}\label{remark:order_function}
        The functions $ l_c(x) = {\rm min} (c\cdot x, 1)$ furthermore satisfy the property that whenever $ c\geq d $, we have that for all $ x\in[0,1] $, $ l_c(x) \geq l_d(x) $.
    \end{remark}
                      
        A small further technical observation stated in Lemma \ref{lem:steep_rho_flat} is required to then prove Lemmas~\ref{lem:bounding_S0} and~\ref{lem:bounding_Sinfty}, which jointly give rise to Lemma~\ref{lem:variance_in_steepflat_properties}.
        \begin{lemma}\label{lem:steep_rho_flat} 
            Let $ \rho,\sigma $ be two commuting $d$-dimensional states that satisfy ${\rm supp}(\rho) \subseteq {\rm supp}(\sigma)$, and denote the $\epsilon$-steep and -flat approximation of $ \rho $ w.r.t. $ \sigma $ as $ \rho_{\mathrm{st}}^\epsilon $ and $ \rho_{\mathrm{fl}}^\epsilon $, according to Definition \ref{def:steepest_approx} and \ref{def:flattest_approx}, respectively. Then, for $ \epsilon_2 \geq \epsilon_1 \geq 0 $,
            \begin{equation}
                \mathcal{L}_{\rho_{\mathrm{st}}^{\epsilon_2}|\sigma} (x)\geq\mathcal{L}_{\rho_{\mathrm{st}}^{\epsilon_1}|\sigma} (x)\geq\mathcal{L}_{\rho|\sigma} (x)\geq\mathcal{L}_{\rho_{\mathrm{fl}}^{\epsilon_1}|\sigma} (x)\geq \mathcal{L}_{\rho_{\mathrm{fl}}^{\epsilon_2}|\sigma} (x).
            \end{equation}
        \end{lemma}   
        \begin{proof}
            The second and third inequality are special cases of the first and fourth inequality, since $ \epsilon_1 \geq 0 $. Let us first consider the fourth inequality. It is shown in \cite{VanderMeer2017} that the $ \epsilon $-flat construction in Definition \ref{def:flattest_approx} is the unique state within an $ \epsilon $-ball of states around $ \rho $, where for any state $ \rho'\in\mathcal{B}^\epsilon (\rho) $,
            \begin{equation}
                \mathcal{L}_{\rho'|\sigma} (x) \geq \mathcal{L}_{\rho_{\mathrm{fl}}^{\epsilon}|\sigma} (x) , \qquad \forall x\in [0,1].
            \end{equation} 
            Since $ \epsilon_2 \geq \epsilon_1 $ implies that $ \rho_{\mathrm{fl}}^{\epsilon_1} \in \mathcal{B}^{\epsilon_2} (\rho) $, the fourth inequality holds. Lastly, consider the first ineuqality. First, note that $ \rho_{\mathrm{st}}^{\epsilon_1}  $ and $ \rho_{\mathrm{st}}^{\epsilon_2} $ always share the same basis, so we may write them as $ \rho_{\mathrm{st}}^{\epsilon_1} = \sum_{i=1}^d \hat{p}_i^{(1)}\ketbra{i}{i} $ and $ \rho_{\mathrm{st}}^{\epsilon_2} = \sum_{i=1}^d \hat{p}_i^{(2)}\ketbra{i}{i} $ respectively. Furthermore, they have the same relative ordering w.r.t. $ \sigma $. In other words, the discrete points that define the respective Lorenz curves are aligned w.r.t. the x-axis, and therefore their condition reduces to a simple comparison between the cummulative sum of the eigenvalues. Concretely, we want that for all $ k \in \lbrace 1,\cdots, d\rbrace $:
            \begin{equation}
                \sum_{i=1}^k \hat{p}_i^{(2)} \geq   \sum_{i=1}^k \hat{p}_i^{(1)}.
            \end{equation}
            Denoting $ R_1  $ and $ R_2 $ to be the respective indices $ R $ according to the construction of Definition \ref{def:steepest_approx}, using $ \epsilon_1 $ and $ \epsilon_2 $ respectively, note that $ R_2 \leq R_1 $. For the various regimes for the Lorenz curves we therefore have:
            \begin{align}
                k = 1       & \qquad\qquad\hat{p}_1^{(2)} = p_1 + \epsilon_2 \geq   p_1 + \epsilon_1 = \hat{p}_1^{(1)},                                                   \\
                1<k\leq R_2 & \qquad\qquad\sum_{i=1}^k \hat{p}_i^{(2)} = \epsilon_2 + \sum_{i=1}^k p_i \geq \epsilon_1 + \sum_{i=1}^k p_i = \sum_{i=1}^k \hat{p}_i^{(1)}, \\
                k> R_2      & \qquad\qquad \sum_{i=1}^k \hat{p}_i^{(2)} = 1 \geq \sum_{i=1}^k \hat{p}_i^{(1)}.                                                            
            \end{align}
            This finishes the proof.
        \end{proof}

    \begin{lemma}\label{lem:bounding_S0}
        Let $\epsilon \in [0,1]$ and let $\rho, \sigma$ be two commuting $d$-dimensional states that satisfy ${\rm supp}(\rho) \subseteq {\rm supp}(\sigma)$, and denote the $\epsilon$-steep approximation of $ \rho $ w.r.t. $ \sigma $ as $ \rho_{\mathrm{st}}^\epsilon $ according to Definition \ref{def:steepest_approx}. Then, 
        \begin{equation} 
            \mathcal{L}_{\rho_{\mathrm{st}}^\epsilon|\sigma} (x)\geq \ell_{r_{\rm st}}(x),~~~
            r_{\rm st} = 2^{S(\rho \| \sigma) - f_\sigma(\rho,\epsilon)},
        \end{equation}
        where $ f_\sigma(\rho,\epsilon) := \sqrt{V(\rho\| \sigma) \left(\epsilon^{-1}-1\right) }$ and $ l_c(x) = {\rm min} (c\cdot x, 1)$. 
    \end{lemma}
    \begin{proof}
        Using the fact that $ [\rho,\sigma]=0 $, we can decompose the states into their simultaneous eigenbasis as
        \begin{equation}
            \rho = \sum_i p_i \proj{i}, \quad \sigma = \sum_i s_i \proj{i},
        \end{equation} 
        and take the ordering of the eigenbasis such that $p_i/s_i \geq p_{i+1}/s_{i+1}$ for all $i \in \{1, \dots, d-1\}$. Next, given $ \epsilon $, define
        \begin{align}
            \tilde{\imath} & = \max_i \left\lbrace i ~\Bigg| \sum_{j=i}^d p_i \geq \epsilon\right\rbrace, 
        \end{align}
        namely $ \tilde{\imath} $ is the largest index such that the tail-sum of the ordered distribution on $p$ is larger or equal to $ \epsilon $. Also, denote the following tail-sums
        \begin{equation}
            \epsilon^+ = \sum_{j=\tilde{\imath}}^d p_j,\quad
            \epsilon^- = \sum_{j=\tilde{\imath}+1}^d p_j,
        \end{equation}
        where we set $\epsilon^- = 0$ if $\tilde{\imath} = d$. 
        By construction, $\epsilon^- \leq \epsilon \leq \epsilon^+$. 
	    We are now going to make use of the steep state with the smaller parameter $ \epsilon^- $: let $\rho_{\rm st}^{\epsilon^-} = \sum_i \hat{p}_i \proj{i}$ denote the $\epsilon^-$-steep approximation of $\rho$ relative to $\sigma$, with $\hat{p}_i$ explicitly defined in Definition \ref{def:steepest_approx}. By construction, we have that 
        \begin{equation}
            A := \min_{i\in P_+} \left(\frac{\hat{p}_i}{s_i}\right) =  \frac{p_{\tilde{\imath}}}{s_{\tilde{\imath}}},
            \end{equation}
            where $ P_+ = \{i|\hat{p}_i > 0\}$. We can now infer that 
            \begin{equation}
                \mathcal{L}_{\rho_{\rm st}^{\epsilon}|\sigma}(x) \geq \mathcal{L}_{\rho_{\rm st}^{\epsilon^-}|\sigma}(x) \geq  \ell_{A}(x),
            \end{equation}
            where the first inequality follows from Lemma~\ref{lem:steep_rho_flat} together with the fact that $ \epsilon \geq \epsilon^- $, while the second inequality follows from Lemma~\ref{thm:intermediate_flat_state}. Hence, it remains to show that
            \begin{equation}\label{eq:AgeqSfs}
            	A \geq 2^{S(\rho \| \sigma) - f_\sigma(\rho,\epsilon)},
	    \end{equation} 
        
	which implies $S(\rho \| \sigma)- f_\sigma(\rho,\epsilon) \leq \log A$.
	    If $S(\rho \| \sigma) \leq \log A$, then this inequality would hold for any non-negative function $f_\sigma(\rho,\epsilon)$. Otherwise, we can derive the explicit form of $ f_\sigma $ so that Eq.~\eqref{eq:AgeqSfs} holds via Cantelli's inequality. More precisely, consider the real-valued random variable $X$ with sample space $\Omega = \{1, 2, \dots, d\}$ distributed as $\mathrm{Prob}(X = \log(p_i/s_i)) = p_i$. We then have
        \begin{align}
            \epsilon \leq \epsilon^+ = {\rm Prob}\left(X
            \leq \log A \right) & = {\rm Prob}\left(S(\rho \| \sigma)  - X  \geq S(\rho \| \sigma) - \log A \right) \\
                                & \leq \frac{V(\rho \| \sigma)}{V(\rho \| \sigma) + \left[S(\rho \| \sigma) - \log A \right]^2},                 
        \end{align}
	    where the first step follows by definition of $\epsilon^+$. The second step expresses the fact that, by virtue of the fact that we have ordered the state bases by decreasing ratios $p_i/s_i$, $\epsilon^+$ is the total probability, with respect to $X$, that a ratio smaller than or equal to $p_{\tilde{\imath}}/s_{\tilde{\imath}}$ is sampled. In the final step we used Cantelli's inequality (Lemma~\ref{lem:cantelli}) with random variable $-X$ and $\lambda \equiv S(\rho \| \sigma) - \log A$, together with the fact that the mean and variance of $-X$ are given by $-S(\rho\|\sigma)$ and $V(\rho\|\sigma)$, respectively. The claim then follows by a simple re-arrangement of the terms above.
    \end{proof}                  
    \begin{lemma}\label{lem:bounding_Sinfty}
        Let $\epsilon \in [0,1]$ and let $\rho, \sigma$ be two commuting $d$-dimensional states that satisfy ${\rm supp}(\rho) \subseteq {\rm supp}(\sigma)$, and denote the $\epsilon$-flat approximation of $ \rho $ w.r.t. $ \sigma $ as $ \rho_{\mathrm{fl}}^\epsilon $ according to Definition \ref{def:flattest_approx}. Then, 
        \begin{equation}
            \mathcal{L}_{\rho_{\mathrm{fl}}^\epsilon|\sigma} (x) \leq \ell_{r_{\rm fl}}(x),~~~r_{\rm fl} = 2^{S(\rho \| \sigma) +  f_\sigma(\rho,\epsilon)},
        \end{equation}
        where $ f_\sigma(\rho,\epsilon) := \sqrt{V(\rho\| \sigma) \left(\epsilon^{-1}-1\right) }$ and $ l_c(x) = {\rm min} (c\cdot x, 1)$. 
    \end{lemma}
    \begin{proof}
        The proof is similar in structure to Lemma~\ref{lem:bounding_S0}, by using the flat approximation instead of the steep one. We begin as well by writing $\rho = \sum_i p_i \proj{i}$ and $\sigma = \sum_i s_i \proj{i}$ such that $p_i/s_i \geq p_{i+1}/s_{i+1}$ for all $i \in \{1, \dots, d-1\}$. Given $ \epsilon $, define 
            \begin{align}
                \tilde{\imath} & = \min_i \left\lbrace i ~\Bigg| \sum_{j=1}^i p_j \geq \epsilon\right\rbrace, 
            \end{align}
            and set $\epsilon^+ = \sum_{j=1}^{\tilde{\imath}} p_j$. Further, set 
            \begin{equation}
                \epsilon^-  = \sum_{j=1}^{\tilde{\imath}-1} p_j - \frac{p_{\tilde{\imath}}}{s_{\tilde{\imath}}}\sum_{j=1}^{\tilde{\imath}-1} s_j, 
            \end{equation}
            if $\tilde{\imath} > 1$ or $\epsilon^- = 0$ if $\tilde{\imath} = 1$. In either case, we by construction have
            \begin{equation}
                \epsilon^- \leq \sum_{j=1}^{\tilde{\imath}-1} p_j \leq \epsilon \leq \epsilon^+.
            \end{equation}
            Now, let $\rho_{\mathrm{fl}}^{\epsilon^-} = \sum_i \bar{p}_i \proj{i}$ denote the $\epsilon^{-}$-flat approximation to $\rho$ relative to $\sigma$. By definition of the flat approximation, Def.~\ref{def:flattest_approx}, we can see that for our choice of $\epsilon^-$, we have 
            \begin{equation}
                \frac{\bar{p}_i}{s_i} = \frac{p_{\tilde{\imath}}}{s_{\tilde{\imath}}} \geq \frac{\bar{p}_j}{s_j}, \quad i = 1, \dots, \tilde{\imath}, \quad j = \tilde{\imath} + 1, \dots, d.
            \end{equation}
            To see the left equality, note that we have defined  $ \tilde{\imath} $ and $\epsilon^-$ in such a way that, in terms of the notation of Definition \ref{def:flattest_approx}, $M = \tilde{\imath} - 1$. Even for the special case of $ \tilde{\imath} =1$, the above ratio is well-defined. Moreover, the definition of the values $\bar{p}_i$ ensures equality of the ratios $\bar{p}_i/s_i$ for $i \leq M$. The right inequality, on the other hand, is a property of the flat approximation that is proven in \cite{VanderMeer2017}. Together, they imply that 
            \begin{equation}
                \max_{l \in [d]} \frac{\bar{p}_l}{s_l} = \frac{p_{\tilde{\imath}}}{s_{\tilde{\imath}}} =: B
            \end{equation}
        Therefore, by Lemma \ref{lem:steep_rho_flat}, Lemma~\ref{thm:intermediate_flat_state} and Remark~\ref{remark:order_function}, we know that $ \mathcal{L}_{\rho_{\mathrm{fl}}^{\epsilon}|\sigma}(x) \leq  \mathcal{L}_{\rho_{\mathrm{fl}}^{\epsilon^-}|\sigma}(x) \leq  \ell_{B}(x)$. Our goal is then to show that $ B \leq 2^{S(\rho \| \sigma) + f_\sigma(\rho,\epsilon)} $. By positivity of $f_\sigma(\rho,\epsilon)$, this is clearly true whenever $S(\rho \| \sigma) \geq \log B$. In case $S(\rho \| \sigma) < \log B$, we again consider a real-valued random variable $X$ with sample space $\Omega = \{1, 2, \dots, d\}$ distributed as $\mathrm{Prob}(X = \log(p_i/s_i)) = p_i$. We then have
        \begin{align}
            \epsilon \leq \epsilon^+ = {\rm Pr}\left(X \geq \log B \right)
              & = {\rm Pr}\left(X  - S(\rho \| \sigma) \geq \log B  - S(\rho \| \sigma) \right) \\
              & \leq \frac{V(\rho \| \sigma)}{V(\rho \| \sigma) + \left[ \log B  - S(\rho \| \sigma) \right]^2},               
        \end{align}
        where we used Cantelli's inequality with $X$ and $\lambda \equiv \log(B)  - S(\rho \| \sigma) > 0$, in the last step. The claim then follows by re-arranging the terms in the above inequality. 
    \end{proof}

\section{Finite but large i.i.d.~sequences} 
    \label{app:second-order asymptotics}
    In this section, we apply Theorem~\ref{res:suff_transition_relative} to study sufficient conditions for approximate state transitions in the case of $n$ i.i.d.~systems.  {We focus on the case of $\sigma$-majorization, where $\sigma' = \sigma$.} We are interested in the regime where $n$ is large but finite and the error $\epsilon_n$ is constant or goes to zero with $n$, but fulfills $\epsilon_n  n\rightarrow \infty$. In technical terms, $\sqrt{\epsilon_n}$ is a \emph{moderate sequence}~\cite{Chubb2018}. 
    \begin{lemma}\label{lem:secondorderasymp}
        Let $ \rho,\rho'$ and $\sigma $ be density matrices satisfying $ [\rho,\sigma] = [\rho',\sigma] = 0$, and that $ {\rm supp}(\rho), {\rm supp}(\rho') \subseteq {\rm supp}(\sigma) $.  Denote $ S\equiv S(\rho\|\sigma), S'\equiv S(\rho'\|\sigma)$ and $ V\equiv V(\rho\|\sigma), V'\equiv V(\rho'\|\sigma) $, respectively. 
        Let $\epsilon_n>0$ be a sequence of errors such that $\epsilon_n n\rightarrow \infty$. Then
        \begin{equation}\label{keyiid}
            \rho^{\otimes n}\succ_{\sigma^{\otimes n},\epsilon_n} \rho'^{\otimes R n},
        \end{equation} 
        with rate
        \begin{equation}\label{keyR}
            R \geq \frac{S}{S'} - \sqrt{\frac{2-\epsilon_n}{\epsilon_n n}}\cdot g(S,S',V,V',\epsilon_n n) + O\left(\frac{1}{\epsilon_n n}\right),
        \end{equation}
        where
        \begin{equation}
            g(S,S',V,V',\epsilon_n n):=\frac{\sqrt{V}}{S'} - \frac{\sqrt{rV'}}{S'} \sqrt{1+O(1/\sqrt{\epsilon_n n})} .
        \end{equation}
    \end{lemma}
    \begin{proof}
        According to Theorem~\ref{res:suff_transition_relative}, the transition in Eq.~\eqref{keyiid} is possible as long as 
        \begin{equation}
            S-\sqrt{\frac{V(2-\epsilon_n)}{n\epsilon_n}} > RS'+\sqrt{\frac{RV'(2-\epsilon_n)}{\epsilon_n n}}.
        \end{equation}
        Rewriting the above equation in terms of $ M = \sqrt{R} $, and dividing throughout by $ n S' $, while grouping the terms
        \begin{align}\label{eq:k_and_kp}
            k & = \frac{2-\epsilon_n}{\epsilon_n n} \frac{V}{S'^2}, 
            \qquad k' = \frac{2-\epsilon_n}{\epsilon_n n} \frac{V'}{S'^2},
        \end{align}
        the above condition simplifies to
        \begin{align}\label{keysomefurtheriid}
            M^2 + \sqrt{k'} M - \left[r-\sqrt{k}\right] < 0, 
        \end{align}
        where we have denoted $ r = S/S'$ to be the ratio of entropies for the initial and final states. The two roots of this equation give a region $ M\in [M^-, M^+] $ for which state transitions may occur. Since we are interested in a sufficient criteria, the lower bound given by $ M^- $ is mainly of interest. Solving the quadratic equation, we then have
        \begin{equation}
            M > M^- = \frac{1}{2} \left[ -\sqrt{k'}-\sqrt{k'+4(r-\sqrt{k})} \right].
        \end{equation}
        Switching back to $ R = M^2 $, we obtain
        \begin{align}
            R & > \frac{1}{4} \left[ 2k'+4(r-\sqrt{k}) + 2\sqrt{k'}\sqrt{k'+4(r-\sqrt{k})}\right] \\
              & > r -\sqrt{k} + \sqrt{rk'}\sqrt{1+O(1/\sqrt{\epsilon_n n})} + O(1/\epsilon_n n).  
        \end{align}
        where in the third term, a factor of $ O(1/\epsilon_n n) $ has been absorbed into $ O(1/\sqrt{\epsilon_n n}) $. Recalling the definitions of $ k,k' $ concludes the proof.
    \end{proof}
    A result related to Lemma~\ref{lem:secondorderasymp} has been derived in~\cite{Chubb2019,korzekwa2019avoiding} for the case of $ \sigma $ being thermal states. In the large $ n $ limit,  the second-order correction term that has a $ 1/\sqrt{\epsilon_n n} $ dependence vanishes whenever $ k = rk' $ for $ k,k' $ defined in Eq.~\eqref{eq:k_and_kp}. This indicates that a ``resonance" happens whenever
    \begin{equation}\label{eq:irreversibility}
        \frac{V/S}{V'/S'}= 1,
    \end{equation}
    which was observed before in~\cite{Chubb2019,korzekwa2019avoiding}. 
    We should stress, however, that Lemma~\ref{lem:secondorderasymp} is far from providing optimal moderate-deviation bounds for the general state-interconversion problem~\cite{Chubb2019,korzekwa2019avoiding}. This is to be expected, since we use a single-shot result that is not tailored to the particular structures appearing in the i.i.d.~limit for large $n$. However, the above analysis shows that qualitative features may be recovered in a very simple manner by making use of Theorem~\ref{res:suff_transition_relative}.
                            
      \section{Proof of Theorem~\ref{lem:schur_concavity_full}} \label{app:proof_schur_concavity}
         Before moving to the proof, let us first prove the following auxiliary lemma.
    \begin{lemma}\label{lem:MarshallOlkin}
    Let $x, s \in \RR^d$ be $d$-dimensional row vectors with $s > 0$,  {and $ y,s' \in \RR^{d'} $ with $ s' >0 $ respectively}. Furthermore, let $A$ be a $d \times d'$ right stochastic matrix, that is, all entries of $A$ are non-negative and every row sums to $1$. If $y = xA$ and $s' = sA$, then
    \begin{equation}
        F(y\|s') \geq F(x \| s)
    \end{equation}
    for any function of the form $   F(x\|s) = \sum_i s_i g \left(\frac{x_i}{s_i}\right) $
    where $ g $ is a function that is concave over the interval $[\min_i\frac{x_i}{s_i}, \max_i \frac{x_i}{s_i}]$.
\end{lemma}
\begin{proof}
	 {
		This proposition is almost the same as stated in Theorem \ref{thm:Renes}, except that in Theorem \ref{thm:Renes} it holds when the function $ g $ is concave over the entire domain. Nevertheless, one should note that it is sufficient for $ g $ to be concave over the entire interval where it is evaluated. This is because given such a $ g $, one can always construct a continuously differentiatble function $ g' $ such that $ g'=g $ within the said interval, and linear outside the interval. Then, $ g' $ is concave. Furthermore, since $ y $ is related to $ x $ by $ y=xA $, the maximal and minimal (strictly positive) ratios of the two states are contractive under classical channels. In other words, the smallest and largest relative eigenvalues satisfy $ \min_i \frac{y_i}{ {s_i'}} \geq \min_i \frac{x_i}{ {s_i}} $ and $ \max_i \frac{y_i}{ {s_i'}} \leq \max_i \frac{x_i}{ {s_i}} $.  {This follows from Theorem~\ref{thm:Renes} and can be seen, for example, by noting that the maximal and minimal ratios correspond to the initial and final slopes of the Lorenz curves, respectively, so that the above inequalities are implied between Lorenz curves, as stated in that theorem.} Therefore it suffices to have the requirement that $ g $ is concave over $[\min_i\frac{x_i}{s_i}, \max_i \frac{x_i}{s_i}]$.} 
\end{proof}

    We are now in a position to prove Theorem~\ref{lem:schur_concavity_full}. 
    \setcounter{theorem}{11}
    \begin{theorem}[]
    		Let $(\rho, \sigma)$ and $(\rho', \sigma')$ be  {two pairs of commuting quantum states, with $\sigma, \sigma'$ both full-rank.}
    		If  {$(\rho, \sigma) \succ (\rho', \sigma')$}, 
    		then it holds that $M_{s_{\min}}(\rho'\| \sigma') \geq M_{s_{\min}}(\rho\|\sigma)$, where
    		\begin{equation}
    			M_x(\rho \| \sigma) := V(\rho\|\sigma) + \left(\frac{1}{\ln(2)} - \log(x) - S(\rho\|\sigma)\right)^2,
    		\end{equation} 
    		and $s_{\min}$ denotes the smallest eigenvalue of $\sigma$.
    \end{theorem}
    \begin{proof}                                                                                                                   
        Since $[\rho, \sigma] = 0$, we can express the value  {$M_{s_{\min}}(\rho \| \sigma)$} in terms of the eigenvalues of $ \rho = \sum_i p_i \proj{i}$ and $ \sigma = \sum_i s_i \proj{i}$. Writing $a= 1/\ln(2) - \ln(s_{\min})$, this yields
        \begin{align}
            {M_{s_{\min}}(\rho \| \sigma)} = V(\rho\|\sigma) - 2a S(\rho\|\sigma) + S(\rho\|\sigma)^2 + a^2
            &= \tr\left[\rho\log^2\left(\frac{\rho}{\sigma}\right)\right] -  2a\tr\left[\rho\log\left(\frac{\rho}{\sigma}\right)\right] + a^2\\
		& =\sum_i p_i \left[\log^2\left(\frac{p_i}{s_i}\right) - 2a \log \frac{p_i}{s_i}\right] + a^2 \\&=\sum_i s_i f_a\left(\frac{p_i}{s_i}\right) + a^2 =: F(p\|s) + a^2,
        \end{align}
        where 
	    $ f_a(x) := x\cdot\left[\log^2(x) - 2a\log(x)\right] $ and $p = (p_1, \dots, p_d)$ and $s = (s_1, \dots, s_d)$.  {By the assumption that $(\rho, \sigma) \succ (\rho', \sigma')$, there exists a quantum channel that maps $\rho$ to $\rho'$ and $\sigma$ to $\sigma'$. Due to the commutation structure,} 
	    by Lemma~\ref{lem:existence_stochastic_matrix}, this implies that there exists a right stochastic matrix $E$ such that $pE = p'$, with $(p')^T = (p'_1 \dots, p'_d)$, and that $sE =s'$. We can therefore apply Lemma~\ref{lem:MarshallOlkin} to find that $F(p \| s) \leq F(p'\|s')$,  whenever $f_a$ is concave over the interval $[\min_i \frac{p_i}{s_i}, \max_i \frac{p_i}{s_i}]$.  {Here, $F(p'\| s')$ is related to $M_{s_{\min}}(\rho' \| \sigma')$ in just the same way as presented for the initial states above, but now using a corresponding common eigenbasis for $\rho'$ and $\sigma'$.} Now, it is straightforward to check that for our choice of $a=\frac{1}{\ln(2)}- \log(s_{\min})$, $f_a$ is concave over the interval $\left[0,1/s_{\min}\right] \supseteq \left[\min_i \frac{p_i}{s_i}, \max_i \frac{p_i}{s_i} \right]$, since the second derivative $f''_a \leq 0$. This establishes that 
        \begin{equation}
             {M_{s_{\min}}(\rho \| \sigma)} \leq   {M_{s_{\min}}( {\rho'} \| \sigma')}. 
        \end{equation} 
    \end{proof}                 

Remark:  Theorem \ref{lem:schur_concavity_full} is particularly interesting because it hints at the possibility of constructing a new classical entropic measure (as a function of relative entropy and variance) that satisfies the data-processing inequality. Note that this is not yet completely true in Theorem \ref{lem:schur_concavity_full} because the monotone $  {M_{s_{\min}}} $ is dependent on the initial reference state.

    \section{Proof of Result~\ref{res:entropy_single} and Theorem~\ref{thm:relative_local_monotonicity}} 
    \label{app:single_out_entropy}
            
    In this appendix, we prove that local monotonicity singles out von~Neumann entropy and the relative entropy among continuous functions on (pairs of) quantum states.
    We will first develop properties of functions that apply in both settings and then prove the respective results. 
    In both cases, we can view the function $f$ in question as acting on pairs of quantum states over certain subsets of density matrices, since in the case of local monotonicity with respect to maximally mixed states we can simply view $f(\rho)$ as $f(\rho,\sigma)$ with $\sigma$ being the maximally mixed state of the same dimension as $\rho$. 
    More generally, let us consider a set of density matrices $\mc S$ that is:
    \begin{enumerate}
        \item Closed under tensor-products: $\sigma,\sigma'\in \mc S\Rightarrow \sigma\otimes\sigma'\in \mc S$. 
        \item Closed under permutations: $\sigma\otimes \sigma' \in \mc S \Rightarrow \sigma'\otimes \sigma\in \mc S$.
    \end{enumerate}
    Furthermore, for any $\sigma\in \mc S$, denote by $\mc C_\sigma$ the set of channels that leave the state invariant:
    \begin{align}
        C\in \mc C_\sigma \Rightarrow C[\sigma]=\sigma, 
    \end{align}
    and denote by $\mc H_\sigma$ the Hilbert-space on which $\sigma\in\mc S$ is defined. 
    Suppose $\sigma_1,\sigma_2\in \mc S$ and $C\in\mc C_{\sigma_1\otimes\sigma_2}$ is given and consider two states $\rho_i$ on $\mc H_{\sigma_i}$ for $i=1,2$. Then we write
    \begin{align}
        \rho_1' := \tr_{2}[C(\rho_1\otimes \rho_2)],\quad   \rho_2' := \tr_{1}[C(\rho_1\otimes \rho_2)]. 
    \end{align}
    A function $f$ on pairs of quantum states $(\rho,\sigma)$ with $\sigma\in \mc S$ and $\rho\in \mc D(\mc H_\sigma)$ is called \emph{locally monotonic with respect to $\mc S$} if 
    \begin{align}
        f(\rho_1,\sigma_1) + f(\rho_2,\sigma_2) \geq f(\rho_1',\sigma_1) + f(\rho_2',\sigma_2) \label{eq:rel_local_monotone} 
    \end{align}
    for any such pairs $(\rho_i,\sigma_i)$ and channels $C\in \mc C_{\sigma_1\otimes \sigma_2}$. 
    We here use the $\geq$ sign in the definition, since we are interested in results relative to the states in $\mc S$.
                                                
    We can further generalize the definition of local monotonicity. 
    Let $\mc C_{\mc S}$ the set of channels that map states from $\mc S$ to $\mc S$ (provided they are in the domain of the corresponding channel). 
    Then we say that $f$ is \emph{locally monotonic with respect to} $\mc C_{\mc S}$ if for all channels $C\in \mc C_{\mc S}$ and $\sigma_1,\sigma_2\in \mc S$ such that $C(\sigma_1\otimes \sigma_2)=\sigma_1'\otimes\sigma_2'\in \mc S$ we have
    \begin{align}
        f(\rho_1,\sigma_1) + f(\rho_2,\sigma_2) \geq f(\rho_1',\sigma_1') + f(\rho_2',\sigma_2') \label{eq:rel_local_monotone2} 
    \end{align}
    if $\rho_1$ and $\rho_2$ are density matrices on the respective Hilbert-space associated to $\sigma_1$ and $\sigma_2$.
    A function that is locally monotonic with respect to $\mc C_{\mc S}$ is always locally monotonic with respect to $\mc S$.
    In the following, we therefore first prove general properties of functions that are locally monotonic with respect to $\mc S$. 
                                                
    Given a function $f$ that is locally monotonic with respect to $\mc S$,
    we define a function $f'$ as 
    \begin{align}
        f'(\rho,\sigma) := f(\rho,\sigma) - f(\sigma,\sigma). 
    \end{align}
    Then $f'$ is still locally monotonic with respect to $\mc S$, since the terms of the form $f(\sigma,\sigma)$ cancel in the corresponding equation~\eqref{eq:rel_local_monotone}.
    \setcounter{theorem}{34}
    \begin{lemma}\label{lemma:monfirst}
        If $f$ is locally monotonic with respect to $\mc S$, then $f(\rho,\sigma)\geq f(C[\rho],\sigma)$ for any $C\in\mc C_\sigma$ and the same is true for $f'$.
        \begin{proof}
            If $C\in\mc C_\sigma$, then $C\otimes \mathbf 1\in \mc C_{\sigma\otimes \sigma}$. But then, since $C[\sigma]=\sigma$, we have
            \begin{align}
                f(\rho,\sigma) + f(\sigma,\sigma) \geq f(C(\sigma),\sigma) + f(\sigma,\sigma). 
            \end{align}
            Since $f'(\rho,\sigma) - f'(C[\rho],\sigma) = f(\rho,\sigma) - f(C[\rho],\sigma)$, the same is true for $f'$.
        \end{proof}
    \end{lemma}
                                                
    \begin{lemma}\label{lemma:additivity}
        If $f$ is locally monotonic with respect to $\mc S$, then $f'$ is additive under tensor-products: 
        \begin{align}
            f'(\rho_1\otimes \rho_2,\sigma_1\otimes \sigma_2) = f'(\rho_1,\sigma_1) + f'(\rho_2,\sigma_2). 
        \end{align}
        \begin{proof}
            Consider the pairs $(\rho_1\otimes \sigma_2,\sigma_1\otimes \sigma_2)$, $(\sigma_1,\sigma_1)$  and the channel $C\in\mc C_{(\sigma_1\otimes \sigma_2)\otimes \sigma_1}$ that permutes the first subsystem of the first pair with the second system. We then find
            \begin{align}
                f'(\rho_1 \otimes\sigma_2,\sigma_1\otimes \sigma_2) + f'(\sigma_1,\sigma_1) \geq f'(\sigma_1\otimes\sigma_2,\sigma_1\otimes\sigma_2) + f'(\rho_1,\sigma_1). 
            \end{align}
            But since $f'(\sigma_1,\sigma_1)=0$, we find $f'(\rho_1\otimes\sigma_2,\sigma_1\otimes \sigma_2) \geq f'(\rho_1 ,\sigma_1)$. The permutation channel is reversible and by considering the reverse transition, we find the converse relation. Hence
            \begin{align}
                f'(\rho_1\otimes\sigma_2,\sigma_1\otimes \sigma_2) = f'(\rho_1,\sigma_1),\quad \forall \sigma_2\in \mc S. 
            \end{align}
            Similarly, we get 
            \begin{align}
                f'(\sigma_1\otimes\rho_2,\sigma_1\otimes \sigma_2) = f'(\rho_2,\sigma_2),\quad \forall \sigma_1\in \mc S. 
            \end{align}
            Considering now the pairs $(\rho_1\otimes\rho_2,\sigma_1\otimes\sigma_2)$ and $(\sigma_1,\sigma_1)$, we similarly find
            \begin{align}
                f'(\rho_1\otimes \rho_2,\sigma_1\otimes\sigma_2)  + \underbrace{f'(\sigma_1,\sigma_1)}_{=0} & =    f'(\sigma_1\otimes\rho_2,\sigma_1\otimes \sigma_2) + f'(\rho_1,\sigma_1) \\
                                                                                                                 & = f'(\rho_1,\sigma_1) + f'(\rho_2\otimes \sigma_2).                           
            \end{align}
            We thus find that $f'$ is additive under tensor products. 
        \end{proof}
    \end{lemma}
    \begin{lemma}\label{lemma:superadditivity}
        If $f$ is locally monotonic with respect to $\mc S$, then $f'$ is super-additive:
        \begin{align}
            f'(\rho_{12},\sigma_1\otimes \sigma_2) \geq f'(\rho_1,\sigma_1) + f'(\rho_2,\sigma_2). 
        \end{align}
        \begin{proof}
            Consider the pairs $(\rho_{12},\sigma_1\otimes \sigma_2),(\rho_1,\sigma_1)$ and again the channel that swaps the first subsystem of the first pair with the second system as in the proof of the previous lemma. Then
            \begin{align}
                f'(\rho_{12},\sigma_1\otimes \sigma_2) + f'(\rho_1,\sigma_1) & \geq f'(\rho_{1}\otimes \rho_2,\sigma_1\otimes \sigma_2) + f'(\rho_1,\sigma_1) \\
                                                                             & = f'(\rho_1,\sigma_1) + f'(\rho_2,\sigma_2) + f'(\rho_1,\sigma_1),             
            \end{align}
            where we used additivity of $f'$ in the last line. 
        \end{proof}
    \end{lemma}
    To summarize, we have found that if $f$ is locally monotone with respect to $\mc S$, then 
    \begin{align}
        f(\rho,\sigma) = f'(\rho,\sigma) + f(\sigma,\sigma) 
    \end{align}
    with $f'$ being additive and super-additive over tensor-products and monotonic under the channels $\mc C_\sigma$: $f'(\rho,\sigma)\geq f'(C[\rho],\sigma)$.
    Result~\ref{res:entropy_single} now follows as the following corollary by considering as $\mc S$ the set of maximally mixed states. 
    \begin{corollary}
        Let $\mc S$ consist of all maximally mixed states and let $f$ be locally monotonic with respect to $\mc S$ and continuous (for fixed $\sigma$). Then
        \begin{align}
            f(\rho) := \log(d) - f(\rho,\id/d) = a S(\rho) + b_d, 
        \end{align}
        where $a$ is a constant and $b_d$ a constant that only depends on the Hilbert-space dimension $d$ of $\rho$.
        \begin{proof}
            In this case, all unitary channels are included in the set of channels. 
            Note that this in particular implies that $f(U\rho U^\dagger, \id/d) = f(\rho,\id/d)$ since unitary channels are reversible. 
            Lemmas~\ref{lemma:monfirst} --~\ref{lemma:superadditivity} now show that $g'(\rho) := -f'(\rho,\id/d)$ fulfills the conditions of Lemma~9 in~\cite{Boes2018}, which shows that
            \begin{align}
                -f'(\rho,\id/d) = g'(\rho) = a S(\rho) + b_d'. 
            \end{align}
            Since $f'(\id/d,\id/d) = 0$, we have $b_d' = - a\log(d)$. We thus get
            \begin{align}
                f(\rho) & = \log(d) - (f(\rho,\id/d)'+f(\id/d,\id/d)) = \log(d) + a S(\rho) -a\log(d) - f(\id/d,\id/d) \\
                        & = a S(\rho) + b_d.                                                                           
            \end{align}
        \end{proof}
    \end{corollary} 
                                                    
    Let us now consider the setting where $\mc S=\mc F$ is the set of finite-dimensional density matrices of full rank. We make use of the following adaption of Lemma~\ref{lemma:monfirst}.

    \begin{lemma}\label{lemma:monsecond}
	    If $f$ is locally monotonic with respect to $\mc C_{\mc F}$, then $f(\rho,\sigma)\geq f(C[\rho],C[\sigma])$ for any $C\in\mc C_{\mc F}$ and the same is true for $f'$. Moreover, $f(\sigma,\sigma)$ is independent of $\sigma\in \mc F$. 
	\begin{proof}
		If $C_1\in \mc C_{\mc F}$ and $C_2\in \mc C_{\mc F}$, then also $C_1\otimes C_2 \in \mc C_{\mc F}$. In particular, we may choose $C_1=C$ and $C_2=\mathbf 1$ as the identity quantum channel acting on $D(\mc H_\sigma)$. By local monotonicity for the initial state $(\rho\otimes \sigma,\sigma\otimes\sigma)$ we then have
		\begin{align}
			f(\rho,\sigma) + f(\sigma,\sigma) \geq f(C[\rho],C[\sigma])+ f(\sigma,\sigma). 
		\end{align}
		Hence $f(\rho,\sigma)\geq f(C[\rho],C[\sigma])$ since by continuity of $f$, $f(\sigma,\sigma)$ must be finite. 
		But for every two states $\sigma_1,\sigma_2\in \mc F$ there are quantum channels in $\mc C_{\mc F}$ mapping one to the other (namely the constant channels). Therefore 
		\begin{align}
		f(\sigma_1,\sigma_1) \geq f(\sigma_2,\sigma_2)\geq f(\sigma_1,\sigma_1). 
		\end{align}
		Hence $f(\sigma,\sigma)$ is independent of $\sigma\in \mc F$.
		This further implies that $f'(\rho,\sigma)-f'(C[\rho],C[\sigma]) = f(\rho,\sigma) - f(C[\rho],C[\sigma])$.

	\end{proof}
    \end{lemma}
Since the quantum channels that swap Hilbert-space tensor factors do not reduce the rank of any state, they are contained in $\mc C_F$. Therefore, the proofs of Lemmas~\ref{lemma:additivity} and \ref{lemma:superadditivity} transfer verbatim to the present case. 
We can now prove theorem~\ref{thm:relative_local_monotonicity}, which we restate here for completeness:
\setcounter{theorem}{14}
\begin{theorem}
    Let $f$ be a function that is locally monotonic with respect to $\mc C_{\mc F}$ and assume that $\rho\mapsto f(\rho,\sigma)$ is continuous for fixed $\sigma\in \mc F$. Then
    \begin{align}
        f(\rho,\sigma) = a S(\rho \| \sigma) + b, 
    \end{align}
    where $a$ and $b$ are constants. 
\end{theorem}

    \begin{proof}[Proof of theorem~\ref{thm:relative_local_monotonicity}]
	    As before, we define $f'(\rho,\sigma) = f(\rho,\sigma)-f(\sigma,\sigma)$. By Lemma~\ref{lemma:monsecond}, we know that $f(\sigma,\sigma)=b$ is a constant and $f'$ is monotonic under arbitrary quantum channels that map full-rank states to full-rank states (on possibly different Hilbert-spaces):
        \begin{align}
            f'(\rho,\sigma) \geq f'(C(\rho),C(\sigma)). 
        \end{align}
	    Furthermore $f'$ is continuous (by assumption), additive, and super-additive (by the same arguemnts as in Lemmas~\ref{lemma:additivity} and \ref{lemma:superadditivity} as mentioned above). Thus $f'$ is continuous, additive, super-additive and monotonic under quantum channels mapping full-rank states to full-rank states and hence fulfills the conditions of the main result of Ref.~\cite{Wilming2017a}. This implies
        \begin{align}
            f'(\rho,\sigma) = a S(\rho\|\sigma). 
        \end{align}
        Hence
        \begin{align}
		f(\rho,\sigma) = f'(\rho,\sigma) + b = a S(\rho\|\sigma) + b. 
        \end{align}
    \end{proof}

    \section{Relation between R\'enyi entropies and cumulants of surprisal} 
    \label{app:relation_between_renyi_curve_and_cumulants_of_surprisal}

Here, we show the relation between the cumulants of surprisal and the R\'enyi entropies claimed in the concluding section. 
For simplicity, we restrict to the unital case, where we only have to deal with entropies instead of divergences. Similar reasoning applies in the case of divergences, though.
   We first define the \emph{cumulant-generating function}  $K_X: \RR \to \RR$ of a real-valued random variable $X$ as 
    \begin{equation}
        K_X(t) := \log_2(\mathbb{E}(2^{tX})),
    \end{equation}
    whenever $X$ is such that this function exists.
    The $n$-th cumulant for a given $K_X$ is defined as $\kappa^{(n)} := K^{(n)}_X(t)|_{t = 0}$, that is, as the $n$-th derivative of $K(t)$ evaluated at $t=0$. If $K_X$ exists, then it is always infinitely differentiable and so all the cumulants are well-defined.

    In our case, we consider a finite-dimensional, positive density matrix $\rho=\sum_i p_i \proj{i}>0$ and the real random variable $X:=-\log(\rho)$ distributed as $\mathrm{Prob}(-\log(\rho)=-\log(p_i))=p_i$ (by slight abuse of notation).
	Then the cumulant generating function is given in terms of the R\'enyi entropies as
	\begin{equation}
		K_{-\log(\rho)}(t) = \log_2(\mathbb{E}(2^{-t\log(\rho)})) = \log_2(\tr[\rho^{1-t}]) = t S_{1-t}(\rho).
	\end{equation} 
	In particular it exists for $t\in (-\infty,1)$ and the cumulants are well defined. 

    We then have the following:
    \setcounter{theorem}{39}
    \begin{lemma}
	    Let $\rho = \sum_i p_i \proj{i} > 0$ be a positive-definite density operator and let $- \log(\rho)$ denote the surprisal with respect to $\rho$, i.e. the random variable distributed as $\mathrm{Prob}(-\log(\rho) = -\log(p_i)) = p_i$. Then for any $\alpha \in (0,\infty)$,
        \begin{equation}
            S_\alpha(\rho) = \sum_{n=1}^\infty \frac{\kappa^{(n)}}{n!}(1- \alpha)^{n-1},
        \end{equation}
        where the cumulants are defined with respect to $X = -\log(\rho)$.
    \end{lemma}
   \begin{proof}
	 We first note that the R\'enyi entropies are related to the cumulant -generating function with respect to the surprisal as 
\begin{equation}
  \kappa :=  K_{-\log(\rho)}(t)= \log_2(\mathbb{E}(2^{-t\log(\rho)})) = \log_2(\tr[\rho^{1-t}]) = t S_{1-t}(\rho),
\end{equation} 
for $t \in [-\infty, 1]$ (for $t \in \{1,0,-\infty\}$ this follows by continuity of the curve $\alpha \mapsto S_\alpha(\rho)$ \cite{vanErven2014}).  Using this relation and repeatedly applying the product rule we find
\begin{align}
    \kappa^{(1)}&= S_{1-t}(\rho) - t S^{(1)}_{1-t}(\rho) \\
    \kappa^{(2)} &=  -2 S^{(1)}_{1-t}(\rho) + t S^{(2)}_{1-t}(\rho) \\
    \dots \\
  \kappa^{(n)}&= (-1)^{n-1} n S^{(n-1)}_{1-t}(\rho) + (-1)^{n} t S^{(n)}_{1-t}(\rho),
\end{align}
	    where $S^{(n)}_{a}(\rho)$ denotes the $n$-th derivative of the curve $\alpha \mapsto S_\alpha(\rho)$ evaluated at $\alpha = a$. The derivatives $S^{(n)}_{1-t}$ clearly exist for $t\in(-\infty,0)\cup (0,1)$. Lemma~\ref{lemma:renyianalytic} shows $S_\alpha(\rho)$ is also infinitely differentiable at $a=1$ (i.e. $t=0$), so that the derivatives $S^{(n)}_{1}$ are well-defined and bounded.

	   At the point  $ \alpha = 1 $ (i.e. $ t = 0 $), some terms from the above set of equations vanish. Writing out the expression for $ \kappa^{(n+1)}$ and taking $ \alpha = 1 $ (i.e. $ t = 0 $) yields
	    \begin{equation}
	    	S^{(n)}_{1}(\rho) = (-1)^{n}  \frac{\kappa^{(n+1)}}{n+1}.
	    \end{equation} 
    Therefore, the Taylor expansion of $S_\alpha(\rho)$ around $\alpha = 1$ can be rewritten in terms of the cumulants of surprisal as 
\begin{equation}
    S_\alpha(\rho) = \sum_{n=0}^\infty \frac{S^{(n)}_{1}(\rho)}{n!}(\alpha - 1 )^n = \sum_{n=1}^\infty \frac{(-1)^{n-1} \kappa^{(n)}}{n!}(\alpha - 1)^{n-1} = \sum_{n=1}^\infty \frac{\kappa^{(n)}}{n!}(1- \alpha)^{n-1},
\end{equation}
where as usual we define $0! = 1$ and $0^0 = 1$.          
    \end{proof}
What is left, is to show that the function $\alpha\mapsto S_\alpha(\rho)$ is infinitely differentiable at $\alpha=1$ or, more precisely, can be extended over $\alpha=1$ analytically. We will use complex analysis and  denote by $\dot B_\epsilon$ an open ball of radius $\epsilon$ around $z=1$ in the complex plane. We sincerely thank Tony Metger and Raban Iten for the proof idea for the following Lemma.   
\begin{lemma}\label{lemma:renyianalytic}
	For any finite-dimensional quantum state $\rho$, there exists an $\epsilon>0$ such that the function $z\mapsto S_z(\rho)$ is holomorphically extendable over $z=1$ on $\dot B_\epsilon\subset \mathbb C$, i.e. there exists a holomorphic function $g(z)$ on $\dot B_\epsilon$ that coincides with $z\mapsto S_z(\rho)$ on $\dot B_\epsilon \setminus\{1\}$.
\begin{proof}
We write
	\begin{align}
		S_z(\rho) = \frac{1}{1-z}\log Q_z,\quad Q_z = \Tr[\rho^z] = \sum_i p_i^z, 
	\end{align}
	where the $p_i$ denote the eigenvalues of $\rho$.
The function $Q_z$ is clearly holomorphic. Moreover, for sufficiently small imaginary values of $z$, it is non-zero.
	Thus there exists an $\epsilon>0$  such that $Q_z$ is both holomorphic and non-zero on $\dot B_\epsilon$.
	Therefore there exists a holomorphic branch of the logarithm on $\dot B_\epsilon$, so that $g(z)=\log Q_z$ is holomorphic on $\dot B_\epsilon$.
	Clearly then $S_z(\rho)$ is holomorphic on $\dot B_\epsilon\setminus\{1\}$. We also have that
	\begin{align}
		\lim_{z\rightarrow 1}(z-1)S_z(\rho) = \lim_{z\rightarrow 1} -\log Q_z = 0.
	\end{align}
	Therefore, by Riemann's theorem on removable singularities, $S_z(\rho)$ can be holomorphically extended over the point $z=1$. 
\end{proof}
\end{lemma}
Note that $g(z)$ being holomorphic, it is in particular continuous, so that $g(1) = \lim_{\alpha\rightarrow 1} S_\alpha(\rho) = S(\rho)$. In other words $g(\alpha)=S_\alpha(\rho)$ for all $\alpha \in (1-\epsilon,1+\epsilon)$.

\section{$S_k(\rho)$ for $k=2,\ldots,d$ encode the spectrum of $\rho$}
    \label{app:spectrumRenyis}
    
    \begin{theorem}
	     The eigenvalues of $ \rho $ can be uniquely determined (including multiplicities) from the values $ S_k(\rho) $ for $k=2,\ldots,d$.
   \end{theorem}
\begin{proof}
    To our knowledge the proof sketch for this theorem first appeared as a comment by Steve Flammia on the website mathoverflow~\cite{Flammia2009}, which we expand here for the reader's convenience. 
	Let the eigenvalues of $\rho$ be given by $p_j$ with $j=1,\ldots, d$. Then for $k\geq 2$ we can express the \emph{$k$-th power-sum} of the $p_j$ as
    \begin{align}
        \sum_j p_j^k = \exp((1-k)S_k(\rho)). 
    \end{align}
    By normalization, we always have $\sum_j p_j=1$. So only the power-sums for $k\geq 2$ provide new information.
    The power-sums can be used to recursively compute the \emph{elementary symmetric polynomials} $e_j(p_1,\ldots, p_d)$ for $j=0,\ldots,d$ using the Girard-Newton identities~\cite{wiki:newton} as (with $e_0(p_1,\ldots,p_d)=1$)
    \begin{align}
	    k e_k(p_1,\ldots, p_d)&= \sum_{i=1}^k (-1)^{i-1} e_{k-i}(p_1,\ldots,p_d) \sum_j p_j^i \\
	     &=\sum_{i=1}^k (-1)^{i-1} e_{k-i}(p_1,\ldots,p_d) \exp((1-i)S_i(\rho)).
    \end{align}
    Note that only the knowledge of the power-sums, or equivalently R\'enyi entropies, for $k=2,\ldots, d$ are required to compute the elementary symmetric polynomials.
    Finally, we can express the characteristic polynomial $c_\rho(\lambda)$ of $\rho$ as a sum over the $d$ elementary symmetric polynomials~\cite{wiki:newton}:
	\begin{align}
		c_\rho(\lambda) &= \sum_{k=0}^d (-1)^k e_k(p_1,\ldots,p_d) \lambda^{d-k}.
	\end{align}
	Solving for the roots of the characteristic polynomial then gives us a unique solution, which is the set of eigenvalues $p_j$ of $\rho$ (including multiplicities). 
    To summarize, we have expressed the characteristic polynomial $C_\rho(\lambda)$ of $\rho$ in terms of the R\'enyi-entropies $S_k(\rho)$ for $k=2,\ldots,d$ and solving for the roots of the characteristic polynomial allows us to determine the spectrum of $\rho$. 
\end{proof}

  The above result shows that in principle the R\'enyi entropies $S_k$ with $k\geq 2$ uniquely determine the spectrum of a density matrix.
  In fact the procedure is relatively simple to implement in Wolfram Mathematica. For purely illustrative purposes we therefore include the following code. While it is not particularly numerially stable, it can be seen to work well for small-dimensional matrices by comparing its output with the built-in routine \texttt{Eigenvalues[]}. 
  
    \lstset{frame=tb,
  language=Mathematica,
  aboveskip=3mm,
  belowskip=3mm,
  showstringspaces=false,
  columns=flexible,
  basicstyle={\small\ttfamily},
  numbers=none,
  numberstyle=\tiny\color{gray},
  keywordstyle=\color{blue},
  commentstyle=\it,
  stringstyle=\color{mauve},
  breaklines=true,
  breakatwhitespace=true,
  tabsize=3
}
\begin{lstlisting}
(* Computes k-th power-sums of a square matrix A. For a density matrix A, this is equivalent to defining the following function as Exp[(1-k)S_k[A]] for k>= 2 if the k-th Renyi entropy of A is given by S_k[A] and Tr[A] for k=1.*)
PowerSum[k_, A_] := Tr[MatrixPower[A, k]];

(* Computes k-th elementary symmetric polynomial of matrix A from power-sums. *)
e[k_, A_] := If[k == 0, 1, 
                 1/k Sum[(-1)^(i-1) e[k-i, A] PowerSum[i, A], {i, 1, k}]
             ];

(* Constructs characteristic polynomial from elmentary symmetric polynomials. *)
CharPol[A_] := Module[{d},
                      d = Length[A];
                      Sum[(-1)^k e[k, A] x^(d-k), {k, 0, d}]
	       ];

(* Solve for roots of characteristic polynomial to obtain eigenvalues. Output as list. *)
eigenvalues[A_] :=   x /. {ToRules[NRoots[CharPol[A] == 0, x]]} // Flatten;
\end{lstlisting}
   
%

\end{document}